\documentclass[10pt,journal,a4paper]{IEEEtran}
\pdfoutput=1

\usepackage{graphicx}
\usepackage{amsfonts}
\usepackage{algorithmic}
\usepackage{algorithm}
\usepackage{amssymb}
\usepackage{bbm}
\usepackage{color}
\usepackage{cite}
\usepackage{amsfonts}
\usepackage{subfigure}

\usepackage[cmex10]{amsmath}
\usepackage{graphicx}
\usepackage{array}
\usepackage{amssymb}
\usepackage{bbm}
\usepackage{mdwmath}
\usepackage{mdwtab}
\usepackage{fixltx2e}
\usepackage{url}
\usepackage[center,scriptsize]{caption}
\newtheorem{theorem}{Theorem}
\newtheorem{proof}{Proof}
\newtheorem{lemma}{Lemma}

\newtheorem{claim}{Claim}

\usepackage{amsmath,amssymb,mathrsfs}
\usepackage{epsfig,epstopdf,epsf,subfigure,graphicx,graphics,subfigure}
\usepackage{url}
\usepackage{color}

\newcommand{\Out}{\mathsf{out}}

\newcommand{\QMF}{\mathrm{QMF}}
\newcommand{\DF}{\mathrm{DF}}
\newcommand{\DDF}{\mathrm{DDF}}
\newcommand{\HYB}{\mathrm{HYB}}
\newcommand{\ACH}{\mathrm{ACH}}
\newcommand{\cut}{\mathrm{cut}}

\newcommand{\lp}{\left(}
\newcommand{\rp}{\right)}
\newcommand{\lb}{\left[}
\newcommand{\rb}{\right]}
\newcommand{\lbp}{\left\{}
\newcommand{\rbp}{\right\}}

\newcommand{\ol}{\overline}
\newcommand{\mcal}{\mathcal}

\newcommand{\what}{\widehat}

\newcommand{\msf}{\mathsf}

\newcommand{\ra}{\rightarrow}

\allowdisplaybreaks
\IEEEoverridecommandlockouts

\title{Cooperative Relaying at Finite SNR -- Role of Quantize-Map-and-Forward}
\author{
\authorblockN{Ayan Sengupta, I-Hsiang Wang, Christina Fragouli}\\
\authorblockA{
\'{E}cole Polytechnique F\'{e}d\'{e}rale de Lausanne (EPFL), Lausanne, Switzerland\\
Email: \{ayan.sengupta, i-hsiang.wang, christina.fragouli\}@epfl.ch
}
\thanks{The work was supported by the EU project CONECT FP7-ICT-2009-257616 and the ERC Starting Grant project NOWIRE ERC-2009-StG-240317.}
}
\begin{document}
\maketitle
\begin{abstract}
Quantize-Map-and-Forward (QMF) relaying has been shown to achieve the optimal diversity-multiplexing trade-off (DMT) \cite{ZhengTse_03} for arbitrary slow fading full-duplex networks \cite{AvestimehrDiggavi_09} as well as for the single-relay half-duplex network \cite{PawarAvestimehr_08}. A key reason for this is that quantizing at the noise level suffices to achieve the cut-set bound approximately to within an additive gap, without any requirement of instantaneous channel state information (CSI). However, DMT only captures the high SNR performance and potentially, limited CSI at the relay can improve performance at moderate SNRs. In this work we propose an optimization framework for QMF relaying over slow fading channels. Focusing on vector Gaussian quantizers, we optimize the outage probability for the full-duplex and half-duplex single relay by finding the best quantization level and relay schedule according to the available CSI at the relays. For the $N$-relay diamond network, we derive an universal quantizer that sharpens the additive approximation gap of QMF from the conventional $\Theta(N)$ bits/s/Hz \cite{AvestimehrDiggavi_09} \cite{LimKim_10} to $\Theta(\log(N))$ bits/s/Hz using only network topology information. Analytical solutions to channel-aware optimal quantizers for two-relay and symmetric $N$-relay diamond networks are also derived. In addition, we prove that suitable hybridizations of our optimized QMF schemes with Decode-Forward (DF) or Dynamic DF protocols provide significant finite SNR gains over the individual schemes.
\end{abstract}


\section{Introduction}
\label{sec:introduction}

Cooperative communication as a means to provide reliable high data-rate wireless transmission has been extensively studied for the past several decades, both from an information-theoretic viewpoint as well as from the perspective of communication protocols. However, the capacity characterization of the relay channel \cite{CoverElGamal_79} has remained one of the most fundamental open problems in information theory, not to mention that of general wireless relay networks. Various relaying techniques for single-relay and multiple-relay wireless networks have been proposed and studied over the years, including amplify-and-forward (AF), decode-and-forward (DF) \cite{CoverElGamal_79}, compress-and-forward (CF) \cite{CoverElGamal_79}, compute-and-forward \cite{NazerGastpar_11}, and 
quantize-map-and-forward (QMF) \cite{AvestimehrDiggavi_09} \cite{LimKim_10}.

Among these, QMF was shown to be the first relaying strategy that provably achieves the capacity of arbitrary Gaussian relay networks to within a bounded additive gap, independent of the channel coefficients in the network \cite{AvestimehrDiggavi_09}.
This is referred to as the \emph{universal approximation} property of QMF, and asymptotically, QMF achieves the capacity of Gaussian relay networks as SNRs go to infinity.
As per the original work in \cite{AvestimehrDiggavi_09}, QMF operates on the following principles: each relay quantizes its received signal
at the noise level and 
maps it onto the transmit codebook; the destination performs a joint decoding operation that together utilizes all the information it receives from the source and the relays to converge on the transmitted source codeword. Unlike conventional CF \cite{CoverElGamal_79} however, explicit decoding of the quantized values from the relay is not a requirement for QMF. Also, the QMF relays in \cite{AvestimehrDiggavi_09} do not require any channel state or network topology information for their operation. Owing to the universal approximation property, QMF is naturally robust to channel fades and hence, a potentially attractive candidate for wireless settings.
In particular, for slow fading wireless relay networks, QMF achieves the optimal diversity-multiplexing tradeoff (DMT) \cite{ZhengTse_03} for arbitrary full-duplex networks \cite{AvestimehrDiggavi_09} and the half-duplex single-relay network \cite{PawarAvestimehr_08}.

However, in spite of possessing the above mentioned attributes, an obvious concern with QMF operation in moderate-SNRs is the additive gap from the information-theoretic cutset upper bound -- $15N$ bits/sec/Hz in \cite{AvestimehrDiggavi_09} improved to $1.26N$ in \cite{LimKim_10}, with $N$ being the number of nodes in the single-source single-destination relay network. Even for reasonably sized networks, these gaps can \emph{a priori} seem prohibitively large for network operation at finite SNRs and existing schemes could potentially outperform QMF at these SNRs even without possessing an universal approximation property. Opposing this concern on the other hand, is the observation that unlike the original works on QMF \cite{AvestimehrDiggavi_09} \cite{LimKim_10}, one could make use of the available channel state information (CSI) at the network nodes, as well as the network topology to boost QMF performance at practically operable SNRs, and translate the inherent advantages of QMF to these regimes.

In this work, we are interested in optimizing the finite-SNR performance of QMF relaying for specific small networks in slow fading scenarios, and the performance metric we take is the outage probability $P_\Out (R)$ at a given rate $R$, which denotes the probability of the system failing to support transmission at rate $R$. Our goal is to examine whether QMF can prove to be a viable alternative or even a scheme of choice in typical operating SNRs when limited channel knowledge is available at the network nodes. We focus on the single-relay network (Fig.~\ref{fig:single_relay}) and the $N$-relay diamond network (a topology introduced in \cite{Schein} in which the source communicates to the destination via $N$ non-interfering relays, Fig.~\ref{N_Relay_Diamond}).
Given (possibly limited) CSI at the relays, we develop an optimization framework to choose the system parameters at the relays, including the quantization level and the half-duplex relay schedule, so that the outage probability $P_\Out (R)$ at a given rate $R$ is minimized.
The two CSI scenarios primarily considered in this work are: (1) the \emph{global CSI} scenario, where all relay nodes have full knowledge of the magnitudes of the fading coefficients in the entire network together with the phase of its incoming link, and (2) the \emph{receiver CSI (CSIR)} scenario, where each relay node only knows the realization (magnitude and phase) of the fading channel coefficient of its incoming link.

We also explore the benefit of incorporating relay decoding into our schemes by using limited CSI to decide when the relay should switch modes. For full-duplex networks, if channel conditions permit, decoding is a natural choice as it cleans up incoming noise at the relay. For half-duplex networks, the decisions are more involved as we have to deal with \emph{dynamic} versions of both the DF and QMF protocols that depend on the listening time (schedule) of the relay. For instance, a schedule that permits decoding may not be the optimal schedule to use. We discuss how our framework deals with such choices in the later sections.


\subsection*{Contributions}
Our main contributions in this paper are summarized in the following:
\begin{itemize}

\item[(i)] For the single-relay full-duplex network, we derive outage-optimal quantizers under various levels of limited CSI. Numerical evaluations demonstrate finite-SNR gains (over the baseline noise-level scheme) ranging from $3$ dB for CSIR settings, to $6$ dB for global CSI-aware relays. For the half-duplex relay, we \emph{jointly} optimize the quantizer and relay schedule under similar settings. The results were partially presented in our conference paper \cite{SWF_Allerton}.

\item[(ii)] For the $N$-relay diamond network, we show that a topology-aware universal quantizer can reduce the additive approximation gap of QMF from the conventional $\Theta(N)$ \cite{AvestimehrDiggavi_09} \cite{LimKim_10} to within $\Theta(\log(N))$ bits/s/Hz, improving the performance of QMF in slow fading networks when no CSI is available. We also derive analytical solutions for channel-aware optimal quantizers for the $2$-relay and symmetric $N$-relay networks when global CSI is available at the relays. These results were presented in part in our conference paper \cite{SWF_ITW}\footnote{A version of the $\Theta(\log(N))$ gap result was also presented in \cite{Chern_logN} independently at the same conference as our work \cite{SWF_ITW}}.

\item[(iii)] In the CSIR settings, we prove that suitable hybridizations of optimized QMF with DF (Dynamic DF for half-duplex) universally outperform 
the outage performance of the individual schemes for the single relay network. For the diamond network, hybridization is also shown to improve the performance of QMF, with simulations demonstrating a $5$ dB gain over QMF alone, and a gain in diversity order over DF in asymmetric diamond networks. 
\end{itemize}

\subsection*{Related Work}
Cooperative relaying over slow fading networks have been extensively studied in the literature. For full-duplex networks, the optimal DMT is found for the single relay channel \cite{YukselErkip_07} and then for arbitrary relay networks \cite{AvestimehrDiggavi_09}. For half-duplex networks, \cite{Azarian_05} proposed the DDF scheme, which is later slightly improved by superposition coding in \cite{PrasadVaranasi_10}. In \cite{PawarAvestimehr_08}, it is shown that QMF relaying achieves the optimal DMT for the single-antenna relay channel. For finite-SNR outage performance, most works in the literature were focused on DF and AF, eg., \cite{BletsasShin_07} --  \cite{YuLee_12}.
The outage performance of QMF over the half-duplex relay channel is recently studied in \cite{YaoKim_13} but the optimization of schedules and performance of dynamic relaying protocols are not further pursued. The idea of hybridization of QMF and DF has been reported in an independent archive submission \cite{HouKramer_13}, where the treatment is limited to full-duplex relay networks.

The paper is organized as follows: Section \ref{sec:model_and_metrics} details the communication model for the single relay, as well as the $N$-relay diamond network and also outlines the metrics that are used to demonstrate the performance benefits of our schemes. Section \ref{sec:relaying_schemes_and_bounds} collects the mathematical expressions for the achievable QMF rates with Gaussian vector quantizers, as well as those for DF, DDF and cutset bound expressions that will be used throughout the paper. Section \ref{sec:full_duplex} details the limited-CSI QMF optimizations for the full-duplex single relay network, while Section \ref{sec:half_duplex} deals with the more involved joint optimization problem for the half-duplex variant. In Section \ref{sec:multiple_relays}, we focus on the optimization choices for the $N$-relay diamond network. Each of the sections \ref{sec:full_duplex}, \ref{sec:half_duplex} and \ref{sec:multiple_relays} contain numerical evaluations that provide a visual insight into our results.

\section{Model and Metrics}
\label{sec:model_and_metrics}

\begin{figure}[!h]
\centering
\subfigure[The single relay network]{
\includegraphics[width=0.5\columnwidth]{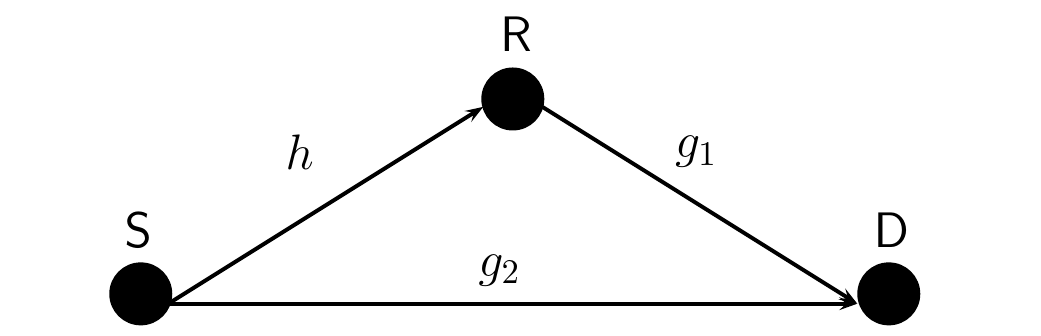}
\label{fig:single_relay}
}
\subfigure[The $N$-Relay Diamond Network with a partition $\Omega = \{1,2\}$ on relays]{
\includegraphics[width=0.4\columnwidth]{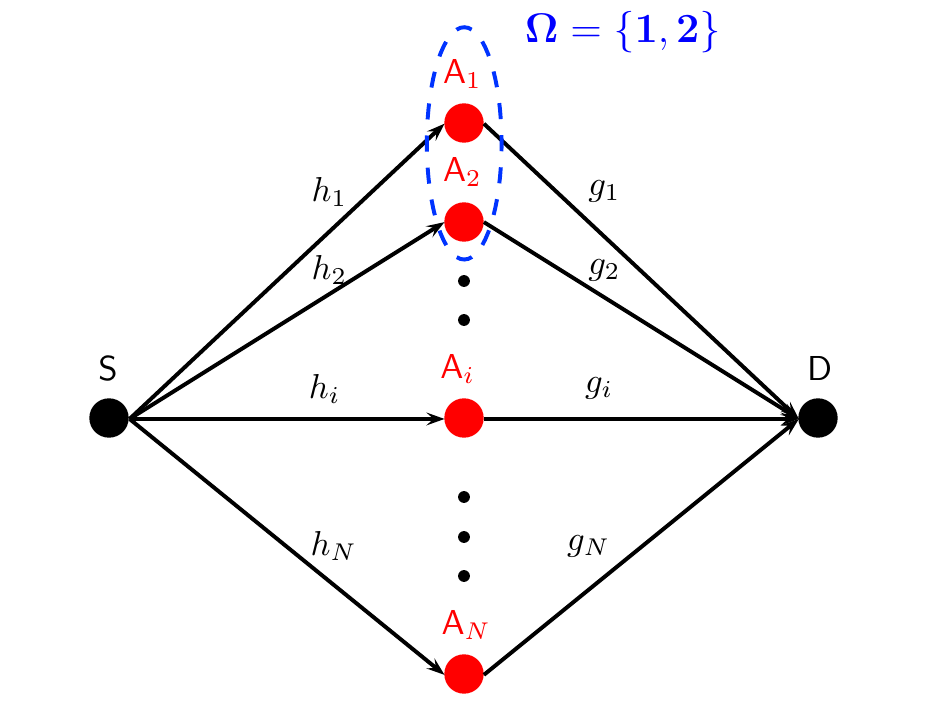}
\label{N_Relay_Diamond}
}
\label{fig:models}
\end{figure}

\subsection{Single Relay Network}

The single relay channel is depicted in Fig. \ref{fig:single_relay}. The source $\mathsf{S}$ communicates with the destination $\mathsf{D}$ with the help of a relay $\mathsf{R}$ that is capable of causal signal processing. The signals transmitted by the source and the relay are denoted by $X$ and $X_r$ respectively. The received signal at the destination and the relay are denoted by $Y$ and $Y_r$ respectively.

In the full-duplex case, the received signals as a function of the transmitted signals are as follows: for $t\in\{1,2,\ldots, T\}$ where $T$ is the communication block length,
\begin{align}\label{eq_Channel}
   Y_r[t] = \msf{h} X[t] + Z_r[t],\
   Y[t] = \msf{g}_1X_r[t] + \msf{g}_2 X[t]  + Z[t].	
\end{align}
The complex channel coefficient from the source to relay is denoted by $\msf{h}$ and those from relay to destination and source to destination are denoted by $\msf{g}_1$ and $\msf{g}_2$ respectively. Additive white Gaussian noises $Z$ and $Z_r$ are i.i.d $\mathcal{CN}(0,1)$ random variables. The transmitted signals are normalized to have an average power constraint of unity at the source and the relay, i.e., $\mathbb{E}(|X|^2) \leq 1$ and $\mathbb{E}(|X_r|^2) \leq 1$. For notational convenience, denote the amplitude of the channel coefficients as follows: $h:=|\msf{h}|, g_1:=|\msf{g}_1|, g_2:=|\msf{g}_2|$.

In the half-duplex scenario, the relay listens for a fraction $f \in [0,1]$ of the total communication block length $N$ and transmits for the fraction $(1-f)$. Without loss of generality, we assume that the relay is in the listening phase during the first $\lfloor fT \rceil$ symbols times ($\lfloor \cdot \rceil$ denotes the closest rounded integer) and in the transmitting phase during the rest. The received signals as a function of the transmitted signals are defined as in \eqref{eq_Channel} except that due to the half-duplex constraint, $X_r[t] = 0$ for $t\in\{1,2,\ldots, \lfloor fT\rceil \}$ when the relay is listening, and for the remaining time $Y_r[t] = 0$, $t\in\{\lfloor fT\rceil +1,\ldots,T\}$.

\subsection{Diamond Networks}

The $N$-Relay diamond network we consider is shown in Fig.~\ref{N_Relay_Diamond}. In this model, the source $\msf{S}$ communicates with the destination $\msf{D}$ via $N$ relay nodes capable of causal signal processing. The signals transmitted by the source $\msf{S}$ and the relay nodes $\msf{A}_i \in \{\msf{A}_1, \msf{A}_2, \ldots, \msf{A}_N\}$ are denoted by $X$ and $X_i$. The received signal at the destination and the relay nodes are denoted by $Y$ and $Y_i$, $i \in \{1,2,\ldots,N\}$. The received signals as a function of the transmitted signals are as follows:
$$
   Y_i = \msf{h}_i X + Z_i,\
   Y = \sum_{i=1}^{N} \msf{g}_iX_i + Z	
$$
where $\msf{h}_i$ represents the complex channel coefficient between $\msf{S}$ and $\msf{A}_i$ and $\msf{g}_i$ denotes that between $\msf{A}_i$ and $\msf{D}$. $Z_i$, $i \in \{1,2,\ldots,N\}$ and $Z$ are i.i.d $\mathcal{CN}(0,1)$ random variables. The transmitted signals are normalized to have an average power constraint of unity at the source and the relays, i.e., $\mathbb{E}(|X|^2) \leq 1$ and $\mathbb{E}(|X_i|^2) \leq 1$, $i \in \{1,2,\ldots,N\}$. Again, for notational convenience, we let $h_i:=|\msf{h}_i|, g_i:=|\msf{g}_i|$ for $i=\{1,2,\ldots,N\}$.

\subsection{Outage Performance in Slow Fading Environment}
In the slow fading scenario, the block-fading channel gains are Rayleigh distributed and are independent from link to link, and across blocks.

The outage probability $P_\Out(R)$ is the probability that the system cannot support the transmission rate $R$ from the source. From the static channel capacity bounds, we have the following lower and upper bounds on the outage probability:
\begin{align*}
&P_\Out(R) \ge P_{\Out,\cut}(R) :=
 \Pr\lbp R > \ol{C}_\cut \rbp \\
&P_\Out(R) \le P_{\Out,\ACH}(R) :=
 \Pr\lbp R > R_{\ACH} \rbp
\end{align*}
\normalsize
where $\ol{C}_\cut$ denotes the information-theoretic cutset upper bound on network capacity and ACH denotes any achievable scheme, like QMF, DF or DDF.

It is shown as a corollary in \cite{AvestimehrDiggavi_09} that QMF achieves the optimal DMT \cite{ZhengTse_03} for arbitrary full-duplex relay networks.
In \cite{PawarAvestimehr_08}, it was shown that even for the single-relay half-duplex channel, QMF relaying with a noise-level quantizer and a fixed $1/2-1/2$ schedule achieves the optimal DMT. At finite SNR however, one can sharpen the upper bound on the outage probability by optimizing the quantization distortion level (and the schedule $f$, for half-duplex networks) at the relay depending on available CSI.

To characterize the finite SNR outage performance corresponding to different \emph{multiplexing gains} on the DMT curve, we demonstrate the performance of our schemes by plotting the outage probabilities as a function of the SNR of one of the links (the SNRs of the other links scale in proportion to the configuration considered) for different rate scalings, i.e, we plot $P_\Out(R)$ vs SNR, where instead of a fixed rate $R$, we use $R = r\log(\mbox{SNR})$ where $r$ can be thought of as the multiplexing gain. This method of representation allows us to take \emph{finite SNR snapshots} of the DMT for varying multiplexing gains. Also, such a representation serves as a good guideline for most wireless systems that adapt their transmission rates according to channel conditions.

\subsection{Limited CSI constraints}

In this paper, we are primarily concerned with two settings of available CSI:
\begin{itemize}
\item [(i)] \emph{Global CSI}: In this setting, the relay has access to the \emph{magnitudes} of the fading coefficients in the entire network but it does not require the phases of its forward channels in order to carry out optimizations on QMF. This piece of information is available in many wireless systems due to reciprocity of channel magnitudes. However, the relay does need phase knowledge of its incoming link in order to generate the quantization codebooks for instance.
\item [(ii)] \emph{CSIR}: when only receiver CSI (CSIR) is available at the relay (both magnitude and phase information), as is the most practical scenario for typical wireless systems.
\end{itemize}

We derive QMF optimizations and optimal combinations with other schemes for both settings for the single relay network as well as the diamond network. For the single relay network, we consider both the full-duplex as well as the half-duplex problem, whereas for the diamond network, we consider only the full-duplex problem.

While arriving at the CSIR-optimal relaying schemes, we also provide optimizations for the single relay network where the relay has access to \emph{local} CSI, i.e., the instantaneous knowledge of its incoming and outgoing channels.

\section{Rate Expressions and Upper Bounds}
\label{sec:relaying_schemes_and_bounds}

To consolidate the development of our optimization framework, in this section we provide the expressions of the information-theoretic cutset upper bound and achievable rates of the QMF relaying scheme as well as other relaying schemes (DF and DDF) considered in this paper. The explicit derivation of these expressions follows straightforwardly from the literature \cite{AvestimehrDiggavi_09} \cite{LimKim_10} \cite{CoverElGamal_79}. For notational convenience, in the rest of this paper, logarithm is of base $2$.

\subsection{Single Relay Network}

\subsubsection{QMF Achievable Rates}

The achievable rate for QMF relaying is given by the following single-letter characterization: for a given probability distribution $p(X)p(X_r)p(\what{Y}_r|Y_r,X_r)$,
\begin{align*}
&R_{\QMF} = \min\lbp\begin{array}{l} I\lp X,X_r; Y\rp - I\lp Y_r;\what{Y}_r|X,X_r,Y\rp,\\
I\lp X;\what{Y}_r,Y|X_r\rp\end{array}\rbp.
\end{align*}
Here $\what{Y}_r$ denotes the quantizer output of the relay. We pick $X\sim\mcal{CN}(0,1), X_r\sim\mcal{CN}(0,1)$ and a Gaussian vector quantization codebook generated by the single-letter probability distribution
\begin{align}\label{eq_GaussianVQ}
\what{Y}_r = Y_r + \what{Z}_r,\ \what{Z}_r\sim \mcal{CN}\lp0,\Delta \rp,
\end{align}
independent of everything else. The parameter $\Delta$ determines the \emph{quantization distortion level}: the larger $\Delta$ is, the coarser the quantization is \cite{LimKim_10}.
With this choice, the QMF achievable rates are as follows:

{\flushleft \textbf{Full-Duplex:}}
$
R_{\QMF}(\Delta) = \left\lceil \min \lbp I_1, I_2 \rbp \right\rceil^+
$, where
\begin{align}\label{eq:rate_FD_single}
\begin{split}
I_1(\Delta)&=\log \lp 1 + \frac{h^2}{1 + \Delta} + g_2^2 \rp\\
I_2(\Delta)&=\log \lp 1 + g_1^2 + g_2^2 \rp - \log \lp \frac{1 + \Delta}{\Delta} \rp
\end{split}
\end{align}

{\flushleft \textbf{Half-Duplex:}}
$
R_{hd,\QMF}(\Delta; f) = \left\lceil \min \lbp I_{hd,1}, I_{hd,2} \rbp \right\rceil^+
$, where
\begin{align}\label{eq:rate_HD_single}
\begin{split}
I_{hd,1}(\Delta,f) &=f\log \lp 1 + \frac{h^2}{1 + \Delta} + g_2^2 \rp\\
&\quad + \lp 1 - f \rp \log \lp 1 + g_2^2 \rp\\
I_{hd,2}(\Delta,f) &= \lp 1 - f \rp \lp \log \lp 1 + g_1^2 + g_2^2 \rp \rp \\
&\quad + f\lbp \log \lp 1 + g_2^2 \rp - \log \lp \frac{1 + \Delta}{\Delta} \rp \rbp
\end{split}
\end{align}

\subsubsection{DF and DDF Achievable Rates}
{\flushleft \textbf{Full-Duplex:}}
\begin{align*}
R_{\DF} = \max \lbp \log \lp 1 + g_2^2 \rp, \min \lbp\begin{array}{l} \log \lp 1 + h^2 \rp, \\ \log \lp 1 + g_1^2 + g_2^2 \rp \end{array}\rbp \rbp
\end{align*}
DF is also DMT optimal for the single relay full-duplex network, and the DF achievable rate is within 1 bit/sec/Hz of the cutset bound. It is to be noted that in the expression above, we assume that the relays do not know the phases of the forward channels and hence, in the last term there is no coherent combining gain. This assumption is true throughout this paper. Also note that we shall use point-to-point transmission if $h^2 < g_2^2$.

{\flushleft \textbf{Half-Duplex:}}
For the half-duplex network, the dynamic decode-forward (DDF) protocol \cite{Azarian_05} is the appropriate variant of DF. In this protocol, the relay adjusts its schedule (the fraction of time it listens) in accordance with the received channel strength to decide how long it needs to listen for it to be able to decode, or whether at all it can decode the transmission from the source. The achievable rate as a function of the relay schedule is given by
\begin{align*}
&R_{\DDF}(f) = \lbp \log \lp 1 + g_2^2 \rp, R_1(f)\rbp,\ \text{where}\\
&R_1(f) =
\min \lbp\begin{array}{l} f\log \lp 1 + h^2 \rp, \\
\lp 1 - f \rp \log \lp 1 + g_1^2 + g_2^2 \rp + f \log \lp 1 + g_2^2 \rp \end{array}\rbp
\end{align*}
The optimal rate achievable (over all possible schedules) is given by:
$
R^*_{\DDF} = \max_{f\in[0,1]} R_{\DDF}(f)
$.

It is important to note that for the half-duplex network, none of the schemes (including DDF) other than QMF is DMT optimal for the entire range of multiplexing gains \cite{PawarAvestimehr_08},

\subsubsection{Upper Bounds on Capacity}

The cutset upper bound on the capacity of the full-duplex single relay channel is given by $\ol{C}_\cut = \min\lbp C_1, C_2 \rbp$ where
\begin{align*}
C_1 = \log \lp 1 + h^2 + g_2^2 \rp,\
C_2 =  \log \lp 1 + \lp g_1 + g_2 \rp^2 \rp.
\end{align*}

For the half-duplex case, the cutset bound expression for a given schedule $f$ is given by
$\ol{C}_{hd,\cut}(f)= \min\lbp C_{hd,1}(f), C_{hd,2}(f) \rbp$ where
\begin{align}\label{eq:cutset_HD}
\begin{split}
C_{hd,1}(f) &= f\log \lp 1 + h^2 + g_2^2 \rp + \lp 1 - f \rp \log \lp 1 + g_2^2 \rp, \\
C_{hd,2}(f) & = f\log \lp 1 + g_2^2 \rp + \lp 1 - f \rp \log \lp 1 + \lp g_1 + g_2 \rp^2 \rp.
\end{split}
\end{align}

\subsection{Diamond Networks}

We denote by $\Omega$, a partition of the index set $[1:N]:=\{1,2,\ldots,N\}$ of the relay nodes $\{\msf{A}_1,\ldots,\msf{A}_N\}$.
Therefore, for any $\Omega\subseteq [1:N]$, $\{\msf{S}\}\cup\{\msf{A}_i : i\in \Omega\}$ is a \emph{cut} of the network (see Fig.~\ref{N_Relay_Diamond} for illustration). Moreover, let $X_\Omega := \{X_i: i\in\Omega\}$.

\subsubsection{QMF Achievable Rates}
The achievable rate for QMF relaying over the $N$-relay diamond network is evaluated using the single-letter characterization\footnote{The Gaussian version was proved in \cite{Lattice_QMF} using lattice vector quantizers.} in \cite{LimKim_10}.
\begin{align*}
&R_{\QMF} = \max\min_{\Omega}\lbp \begin{array}{l}I\lp X,X_{\Omega}; \what{Y}_{\Omega^c},Y|X_{\Omega^c}\rp\\
-I\lp Y_{\Omega};\what{Y}_{\Omega}|X,X_{[1:N]},\what{Y}_{\Omega^c},Y\rp\end{array}\rbp
\end{align*}
where the maximum is taken over all probability distributions $p(X){\prod_{i\in[1:N]} p\lp X_i\rp p\lp\what{Y}_i|Y_i,X_i\rp}$. Here $\what{Y}_{i}$ denotes the quantizer output of relay $\msf{A}_i$. Generating the Gaussian vector quantization codebooks with the following independent single-letter probability distributions
$
\what{Y}_{i} = Y_i + \what{Z}_i,\ \what{Z}_i\sim \mcal{CN}\lp0,\Delta_i\rp,\ i=1,\ldots, N
$,
the above achievable rate is evaluated as
\begin{align*}
R_{\QMF}\lp  \Delta_{[1:N]}\rp = \min_{\Omega}\lbp R\lp\Omega; \Delta_{[1:N]}\rp\rbp
\end{align*}
where $\Delta_J := \{\Delta_{j}|\ j\in J\}$, and
$R\lp\Omega;\Delta_{[1:N]}\rp$ is as follows:
\begin{align}
&R\lp\Omega; \Delta_{[1:N]}\rp \notag\\
=& \lb\begin{array}{l}\log\lp 1+\sum_{i\in\Omega}|g_i|^2\rp + \log\lp1+\sum_{j\in\Omega^c}\frac{|h_j|^2}{1+\Delta_{j}}\rp\\
- \sum_{i\in\Omega}\log\lp\frac{1+\Delta_i}{\Delta_i}\rp\end{array}\rb^+ \label{eq_QMFRate}
\end{align}

\subsubsection{DF Achievable Rates}
The DF achievable rate is given by:
\begin{align*}
R_{DF} = \max_{\Omega}\lbp \min\lbp \log \lp 1 + \sum_{i \in \Omega} g_i^2 \rp, \min_{i \in \Omega} \log\lp 1 + h_i^2\rp \rbp \rbp
\end{align*}
\subsubsection{Upper Bounds on Capacity}

The cutset upper bound on the capacity of the network is given by \eqref{eq_MinCut} in the following.
\begin{align}
\label{eq_MinCut}&\overline{C}_{\cut} \le
\min_{\Omega}\left\{\log\lp 1+{\left(\sum_{i\in\Omega}g_i\right)}^2\rp + \log\lp1+\sum_{j\in\Omega^c}h_j^2\rp\right\}
\end{align}

\section{Full-Duplex Single Relay Network}
\label{sec:full_duplex}

In this section we focus on the full-duplex single-relay network.
We study the ways in which the full-duplex relay can utilize the available CSI to optimize and improve the performance of QMF in slow fading scenarios.
As mentioned in Section~\ref{sec:relaying_schemes_and_bounds}, we focus on QMF with Gaussian random codebooks and Gaussian vector quantizers as shown in \eqref{eq_GaussianVQ}.
Hence, the optimization parameter that the relay can choose is the quantization level $\Delta$ used in \eqref{eq_GaussianVQ}.

We assume that the relay knows the probability distributions of the three links $\{\msf{h},\msf{g}_1,\msf{g}_2\}$, but has limited access to the realizations of them in each transmission block. Recall that the channels are Rayleigh-faded and are independent from link to link, and across blocks:
$$
\msf{h}\sim \mcal{CN}(0,1/\rho),\ \msf{g}_1\sim \mcal{CN}(0,1/\lambda_1),\ \msf{g}_2\sim \mcal{CN}(0,1/\lambda_2)
$$
where $1/\rho,1/\lambda_1,1/\lambda_2$ denote the signal-to-noise ratios in the Source-Relay link, the Relay-Destination link, and the Source-Destination link respectively.

In the rest of this section, we study the following CSI scenarios:
\begin{itemize}
\item The relay has global CSI, i.e., the instantaneous realizations of $\msf{h}, g_1$ and $g_2$.
\item The relay has local CSI, that is, the instantaneous realizations of $\msf{h}$ and $g_1$.
\item The relay has CSI at the receiver (CSIR) only, that is the realization of $\msf{h}$.
\end{itemize}

In the global CSI setting, the optimization problem is equivalent to \emph{rate maximization}. In the other cases, we formulate the objective function as the outage probability, \emph{conditioned} on the known channel gain realizations at the relay. For the CSIR-limited case, we also prove that a hybridization of QMF (using optimized quantizers) with DF provides a strictly better outage performance than either one of the two schemes for all channel configurations.

\subsection{Global CSI at the Relay}

When the relay has access to the instantaneous channel gain realizations $\msf{h}, g_1$ and $g_2$, it will choose a quantizer distortion $\Delta^*$ that maximizes $R_{\QMF}(\Delta)$. Hence using ~\eqref{eq:rate_FD_single}, the optimization problem can be stated as a $\max\min$ problem as follows:
\begin{align*}
R_{\text{QMF}}^*=\max_{\Delta > 0}\lbp \left\lceil \min \lbp I_1, I_2 \rbp \right\rceil^+ \rbp
\end{align*} and the optimizing $\Delta$ is denoted by $\Delta^*$.

We note that $I_1$ is \emph{monotonically decreasing} in $\Delta$ and $I_2$ is \emph{monotonically increasing} in $\Delta$. Also, at the boundary points of the admissible range of $\Delta$, we have,
\begin{align*}
I_1(0) = \lim_{\Delta \rightarrow 0^+} I_1&=\log \lp 1 + h^2 + g_2^2 \rp\\
I_1(\infty) = \lim_{\Delta \rightarrow \infty} I_1&=\log \lp 1 + g_2^2 \rp\\
I_2(0) = \lim_{\Delta \rightarrow 0^+} I_2&=-\infty\\
I_2(\infty) = \lim_{\Delta \rightarrow \infty} I_2&=\log \lp 1 + g_1^2 + g_2^2 \rp
\end{align*}

The properties $I_1 > 0 \quad \forall \Delta>0$, $I_2(0) < I_1(0)$ and $I_2(\infty) > I_1(\infty)$, together with the above mentioned monotonicity properties ensure that the optimum quantizer distortion is given by
\begin{align*}
\Delta^* &= \arg_{\Delta}\lbp I_1 = I_2 \rbp = \frac{1 + h^2 + g_2^2}{g_1^2}
\end{align*}

\subsection{Local CSI at the Relay}

The goal here is to find a quantizer distortion based on the available channel knowledge that minimizes the conditional probability of outage at a given rate $R$, denoted by $P_{\Out|h,g_1}(R; \Delta)$.
Since the fading channel coefficients $\{\msf{h}, \msf{g}_1,\msf{g}_2\}$ are mutually independent, we have
\begin{align*}
&1 - P_{\Out|h,g_1}(R; \Delta) =\Pr \lbp R \leq R_{\QMF}(\Delta) \rbp \\
&=\Pr \lbp \begin{array}{l}R \leq \log \lp 1 + \frac{h^2}{1 + \Delta} + g_2^2 \rp , \\ R \leq \left\lceil \log \lp 1 + g_1^2 + g_2^2 \rp - \log \lp \frac{1 + \Delta}{\Delta} \rp \right\rceil^+ \end{array}\rbp \\
& = \Pr \lbp \begin{array}{l}g_2^2 \geq \left\lceil \underbrace{2^R - \frac{h^2}{1 + \Delta} - 1}_{\beta_1} \right\rceil^+, \\ g_2^2 \geq \left\lceil \underbrace{2^R \lp \frac{1 + \Delta}{\Delta} \rp - g_1^2 - 1}_{\beta_2}\right\rceil^+ \end{array} \rbp\\
& = \Pr \lbp g_2^2 \geq \max \lbp \left\lceil \beta_1 \right\rceil^+, \left\lceil \beta_2 \right\rceil^+ \rbp \rbp \\
& = e^{-\lambda_2 \max \lbp \left\lceil \beta_1 \right\rceil^+, \left\lceil \beta_2 \right\rceil^+ \rbp}
\end{align*}

The problem of finding the optimal $\Delta^*$ that minimizes $P_{\Out|h,g_1}(R; \Delta)$ can then be stated as
\begin{align*}
\Delta^* &= \arg\max_{\Delta > 0} e^{-\lambda_2 \max \lbp \left\lceil \beta_1 \right\rceil^+, \left\lceil \beta_2 \right\rceil^+ \rbp} \\
&= \arg\min_{\Delta > 0}\max \lbp \left\lceil \beta_1 \right\rceil^+, \left\lceil \beta_2 \right\rceil^+ \rbp
\end{align*}

Now, $\left\lceil \beta_1 \right\rceil^+$ is \emph{non-decreasing} in $\Delta$ and $\left\lceil \beta_2 \right\rceil^+$ is \emph{non-increasing} in $\Delta$. The properties that $\left\lceil \beta_2 \right\rceil^+(0) > \left\lceil \beta_1 \right\rceil^+(0)$ and $\left\lceil \beta_2 \right\rceil^+(\infty) < \left\lceil \beta_1 \right\rceil^+(\infty)$, together with the monotonicity conditions ensure that the optimum quantizer distortion is given by
\begin{align*}
\Delta^* &= \arg_{\Delta > 0} \lbp \beta_1 = \beta_2 \rbp \\
& = \arg_{\Delta > 0} \lbp g_1^2 \Delta^2 + \lp g_1^2 - h^2 - 2^R \rp \Delta - 2^R = 0 \rbp \\
& = \frac{\sqrt{\lp g_1^2 - h^2 - 2^R \rp^2 + 4g_1^2 2^R} - \lp g_1^2 - h^2 - 2^R \rp}{2g_1^2}
\end{align*}

\subsection{CSIR at the Relay}
To obtain an outage-optimal quantizer in this setting, the metric to minimize is again the conditional probability of outage, but this time conditioned \emph{only} on $h$. Proceeding in a similar fashion as in the local CSI case, we have,
\begin{align*}
Q(\Delta) &\triangleq 1 - P_{\mbox{out}|h}(R; \Delta)\\
& = \Pr \lbp \begin{array}{l} g_2^2 \geq \left\lceil \underbrace{2^R - \frac{h^2}{1 + \Delta} - 1}_{\alpha_1} \right\rceil^+, \\ g_1^2 + g_2^2 \geq \underbrace{2^R \lp \frac{1 + \Delta}{\Delta} \rp - 1}_{\alpha_2} \end{array} \rbp
\end{align*}
We note that in contrast to the local CSI setting where only $g_2^2$ was treated as a random variable, both $g_1^2$ and $g_2^2$ are to be treated as random variables in the optimization for the CSIR setting. With regards to $\alpha_1$ and $\alpha_2$, note that $\alpha_2 > \alpha_1$ and $\left\lceil \alpha_1 \right\rceil^+ = \alpha_1$ for $\Delta > \Delta_t = \frac{h^2}{2^R - 1} - 1$.

\begin{figure}[!ht]
  \centering
  \includegraphics[width=0.7\columnwidth]{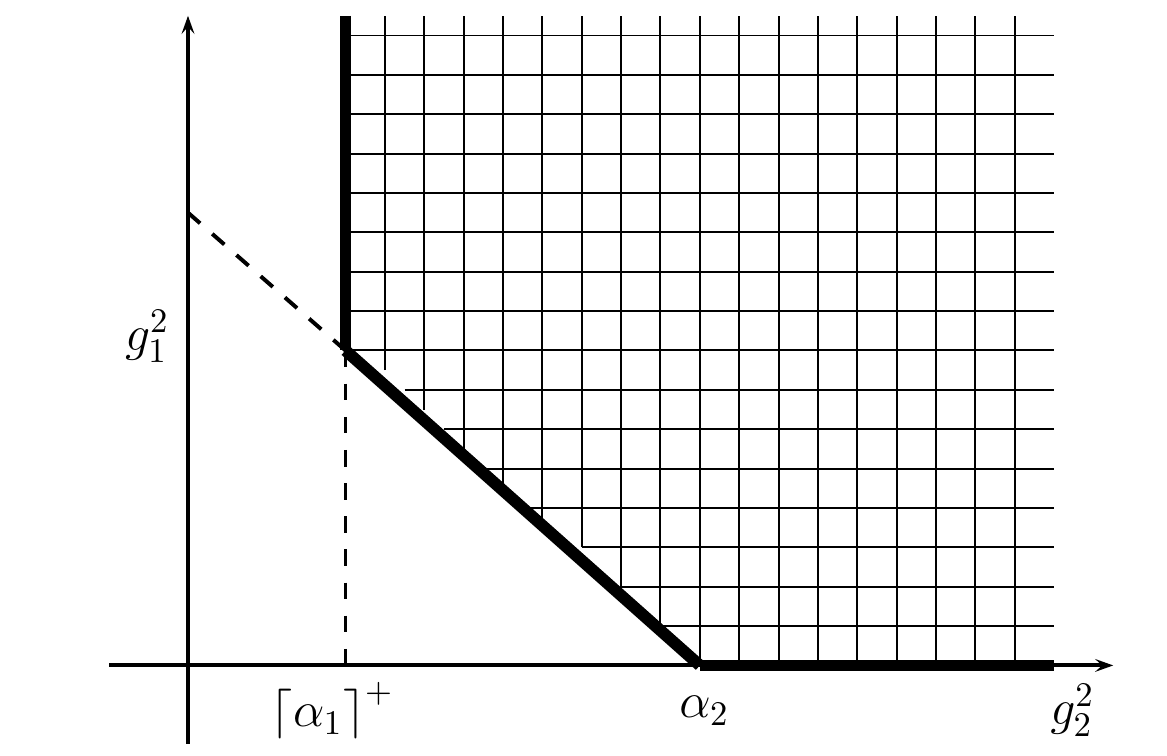}\\
  \caption{Illustration of the integration region for computing $Q(\Delta)$}\label{fig:region_CSIR_fd}
\end{figure}

By integrating the joint density of the two independent exponentially distributed random variables corresponding to $g_1^2$ and $g_2^2$ over the shaded region in the $2$-D plane (as shown in Fig. \ref{fig:region_CSIR_fd}), the above probability is computed as follows
\begin{align*}
&Q(\Delta) = \\
&\begin{cases}
 \frac{\lambda_2}{\lambda_2 - \lambda_1}e^{-\lp \lambda_1 \alpha_2 + (\lambda_2-\lambda_1) \left\lceil \alpha_1 \right\rceil^+ \rp} - \frac{\lambda_1}{\lambda_2 - \lambda_1}e^{- \lp \lambda_2 \alpha_2 \rp} \quad \lambda_1 \neq \lambda_2\\
 e^{-\lambda_1 \alpha_2}\lp 1 + \lambda_1 \alpha_2 - \lambda_1 \left\lceil \alpha_1 \right\rceil^+ \rp \quad \lambda_1 =\lambda_2
\end{cases}
\end{align*}

Hence the problem of choosing the optimal quantizer distortion can be stated as follows:
$$
\Delta^* = \arg\max_{\Delta > 0} Q(\Delta).
$$

For simplicity, we focus on the case where $\lambda_1 = \lambda_2$. The case where $\lambda_1\ne\lambda_2$ is treated in Appendix~\ref{app_asym}. When $\lambda_1 = \lambda_2$,
\begin{equation*}
Q'(\Delta) = \lambda_1 \lp \left\lceil \alpha_1 \right\rceil^+ - \alpha_2 \rp \alpha_2' - {\left\lceil \alpha_1 \right\rceil^+}'
\end{equation*}where all derivatives are with respect to $\Delta$.

To derive the optimal $\Delta^*$, we first assume that $\alpha_1 > 0$ and solve the optimization problem. With this assumption, we note that $Q'(0) = +\infty$, i.e $Q(\Delta)$ is increasing at $\Delta = 0$. Also, the solution to $Q'(\Delta) = 0$ is given by the solution of a cubic equation with \emph{exactly one} positive root (this can easily be seen from the Descartes' sign scheme), which we denote by $\Delta^{\dagger}$
\begin{align*}
\Delta^{\dagger} &= \arg_{\Delta>0}\lbp \begin{array}{l} \lp \frac{h^2}{\lambda_1} \rp\Delta^3 - \lp 2^R \lp 2^R + h^2 \rp \rp \Delta^2 \\ - \lp 2^R \lp 2^{R+1} + h^2 \rp \rp \Delta - \lp 2^{2R} \rp = 0 \end{array}\rbp
\end{align*}where $\Delta^{\dagger}$ can be evaluated analytically using the properties of cubic equations.

Since $Q'(0) = +\infty$, and $\Delta^{\dagger}$ is the \emph{only} critical point of $Q(\Delta)$, the maximizing $\Delta^* = \Delta^{\dagger}$, provided our initial assumption of $\alpha_1 > 0$ is correct, which translates to $\Delta^{\dagger} > \Delta_t = \frac{h^2}{2^R - 1} - 1$ being satisfied.

If however, $\Delta^{\dagger} \leq \Delta_t$, our initial assumption is invalid, and we set $\left\lceil \alpha_1 \right\rceil^+ = 0$ in the $Q(\Delta)$ expression. We also note that in this case, the maximizing $\Delta^*$ will lie in the interval $\lp 0, \Delta_t \rb $. In such a case, $Q'(\Delta) = -\lambda_1\alpha_2\alpha_2' > 0 \quad \forall \Delta \in \lp 0,\Delta_t \rb $, i.e., $Q(\Delta)$ is monotonically increasing in $\Delta$. Hence, the maximizing $\Delta^* = \Delta_t$.

Thus, combining the above cases, we have,
\begin{equation*}
\Delta^* = \max\lbp \Delta^{\dagger}, \Delta_t \rbp \quad \mbox{for} \quad \lambda_1 = \lambda_2
\end{equation*}

\subsection{CSIR limited Hybrid DF/QMF}
\label{subsec:csir_hybrid_fd}

In this strategy, we let the relay perform a DF operation if it is able to decode the codeword sent from the source. However, if it cannot decode, it will go into QMF mode, whereby it will use the CSIR-optimal quantizer distortion to quantize the received signal and subsequently map and forward it to the destination. The relay will also transmit a 1-bit flag to the destination letting it know the relaying scheme it has used, i.e., DF or QMF. We next show that the outage performance of this hybrid scheme (denoted by HYB) is superior to the outage performance of both the DF and CSIR-optimal QMF schemes individually.

We write the outage probabilities at a given rate $R$ for the relaying schemes using the total probability theorem by conditioning on whether the relay can decode the source transmission. Thus, for the relaying schemes DF, QMF and HYB, we have:
\begin{align}
P_{\Out,\DF}(R) &= \begin{array}{l}\lbp\begin{array}{l}\Pr \lbp R > \log \lp 1 + g_1^2 + g_2^2 \rp \rbp. \\ \Pr \lbp R < \log \lp 1 + h^2 \rp \rbp \end{array}\rbp\\
+ \lbp\begin{array}{l}\Pr \lbp R > \log \lp 1 + g_2^2 \rp \rbp. \\ \Pr \lbp R \geq \log \lp 1 + h^2 \rp \rbp \end{array}\rbp \end{array}\label{eq:P_out_DF}\\
P_{\Out,\QMF}(R) &= \begin{array}{l}\lbp\begin{array}{l}\Pr \lbp R > R_{\QMF}(\Delta^*) \rbp. \\ \Pr \lbp R < \log \lp 1 + h^2 \rp \rbp \end{array}\rbp \\
+ \lbp\begin{array}{l}\Pr \lbp R > R_{\QMF}(\Delta^*) \rbp. \\ \Pr \lbp R \geq \log \lp 1 + h^2 \rp \rbp \end{array}\rbp \end{array} \label{eq:P_out_QMF}\\
P_{\Out,\HYB}(R) &= \begin{array}{l}\lbp\begin{array}{l}\Pr \lbp R > \log \lp 1 + g_1^2 + g_2^2 \rp \rbp. \\ \Pr \lbp R < \log \lp 1 + h^2 \rp \rbp \end{array}\rbp\\
+ \lbp\begin{array}{l}\Pr \lbp R > R_{\QMF}(\Delta^*) \rbp. \\ \Pr \lbp R \geq \log \lp 1 + h^2 \rp \rbp \end{array}\rbp \end{array}\label{eq:P_out_HYB}
\end{align}

An important point to note for the above expressions is that in terms involving $R_{\QMF}(\Delta^*)$, the probability must be computed by further conditioning on the source-to-relay channel $h$, and then using the total probability theorem. This is because in CSIR-optimized QMF, the distortion $\Delta^*$ depends on the channel realization $h$.

To show that $P_{\Out,\HYB}(R) < P_{\Out,\QMF}(R)$, from equations \eqref{eq:P_out_QMF} and \eqref{eq:P_out_HYB} we have to show:
\begin{align*}
\Pr \lbp R > \log \lp 1 + g_1^2 + g_2^2 \rp \rbp < \Pr \lbp R > R_{\QMF}(\Delta^*) \rbp
\end{align*}
where $\Delta^*$ is the CSIR-optimal quantizer distortion derived previously.

From the QMF rate expression in \eqref{eq:rate_FD_single}, we see that
\begin{align*}
&R_{\QMF}(\Delta^*) < \log \lp 1 + g_1^2 + g_2^2 \rp - \log \lp \frac{1 + \Delta^*}{\Delta^*} \rp\\
\implies &R_{\QMF}(\Delta^*) < \log \lp 1 + g_1^2 + g_2^2 \rp
\end{align*}
which proves that $P_{\Out,\HYB}(R) < P_{\Out,\QMF}(R)$.

Next, to show that $P_{\Out,\HYB}(R) < P_{\Out,\DF}(R)$, from equations \eqref{eq:P_out_DF} and \eqref{eq:P_out_HYB} we have to show:
\begin{align*}
\Pr \lbp R > R_{\QMF}(\Delta^*) \rbp < \Pr \lbp R > \log \lp 1 + g_2^2 \rp \rbp.
\end{align*}

From the definition of $\Delta^*$, the CSIR-optimal quantizer distortion, we know that it minimizes the outage probability $\Pr \lbp R > R_{\QMF}(\Delta) \rbp$ for every realization $h$ of the source to relay channel. Hence, if we pick any $\Delta = \Delta_1$ irrespective of $h$ (in other words, we pick the same quantizer distortion $\Delta_1$ for all realizations $h$ instead of using optimized distortions for every realization), we have:
$
\Pr \lbp R > R_{\QMF}(\Delta^*) \rbp < \Pr \lbp R > R_{\QMF}(\Delta_1) \rbp
$.
Picking $\Delta_1 = \infty $ gives us:
\begin{align*}
\Pr \lbp R > R_{\QMF}(\Delta^*) \rbp < \Pr \lbp R > \log \lp 1 + g_2^2 \rp \rbp
\end{align*}which proves that $P_{\Out,\HYB}(R) < P_{\Out,\DF}(R)$.

\subsection{Numerical Evaluation}

\begin{figure}[!h]
\centering
\subfigure[All channels are i.i.d. The rate scales as: $R = 0.3\log_2(\mathsf{SNR})$]{
\includegraphics[scale=0.4]{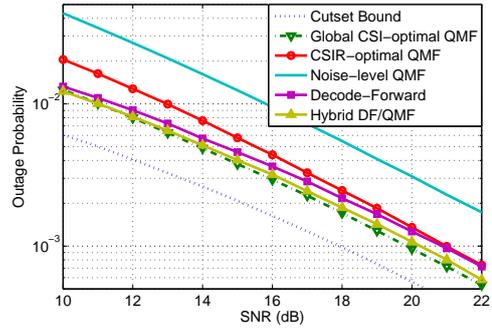}
\label{res:FD_0.3_iid}
}
\subfigure[$\mathbb{E}\{|g_1|^2\} = \mathbb{E}\{|g_2|^2\} = 10\mathbb{E}\{|h|^2\}$. $x$-axis : $\mathbb{E}\{|h|^2\}$. The rate scales as $R = 0.7\log_2(\mathsf{SNR})$
]{
\includegraphics[scale=0.4]{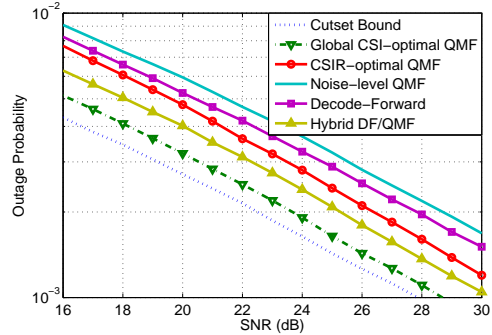}
\label{res:FD_0.7_weakStoR}
}
\caption[]{Outage Performance in Rayleigh faded channels over the full-duplex single-relay network}
\label{res:singlerelay_FD}
\end{figure}

In Fig. \ref{res:singlerelay_FD}, the outage performance of the relaying schemes are plotted for different rate scalings and channel configurations. In the setting shown in Fig. \ref{res:FD_0.3_iid} with i.i.d. channels and a multiplexing gain of $0.3$, the CSIR-optimized QMF is shown to provide a $3$ dB gain over the baseline noise-level scheme. The global CSI optimal QMF roughly offers an additional $2$ dB gain over the CSIR-optimal QMF. Similar optimization gains are also observed for the setting in Fig. \ref{res:FD_0.7_weakStoR} with a higher rate scaling and more asymmetric channels.

Another interesting point observed from simulations is that the performance of the \emph{local} CSI-optimal QMF is numerically indistinguishable from that of the global CSI variant for all channel configurations we tested. This possibly indicates that the $\msf{S}-\msf{D}$ direct link does not influence the performance of QMF significantly as far as relay parameter choices are concerned.

From Fig. \ref{res:FD_0.3_iid}, we also see that for i.i.d. channels, in CSIR limited settings, DF marginally outperforms CSIR-optimal QMF (by about $1$ dB), but the (CSIR limited) hybrid scheme outlined in section \ref{subsec:csir_hybrid_fd} outperforms both DF and CSIR-optimal QMF. For the setting shown in Fig. \ref{res:FD_0.7_weakStoR} on the other hand, where the source-to-relay link is weaker in comparison to the other links, CSIR-optimal QMF outperforms DF by about $2$ dB. The hybrid again outperforms them both, providing an additional $1$ dB benefit over CSIR-optimal QMF.

\section{Half-Duplex Single Relay network}
\label{sec:half_duplex}

To extend our framework to half-duplex QMF relaying, we have two parameters -- the schedule (fraction of listening time) $f$, and the quantizer distortion $\Delta$, that we have to \emph{jointly} optimize depending on the CSI available to the relay. This problem is significantly more involved than the full-duplex QMF problem. In fact it turns out that it is difficult to provide clean analytical characterizations for most of the CSI conditions that we have tackled in the paper.

\subsection{Global CSI at the Relay:}
\label{subsec:globalCSI_HD}
With global CSI, the equivalent rate-maximization problem can be stated using \eqref{eq:rate_HD_single} as:
\begin{align*}
R_{\QMF}^*=\max_{f \in (0,1);\Delta > 0}\lbp \left\lceil \min \lbp I_{hd,1}, I_{hd,2} \rbp \right\rceil^+ \rbp .
\end{align*}

From the expressions above, it is evident that
\begin{itemize}
\item $I_{hd,1}$ is \emph{monotonically increasing} in $f$ and \emph{monotonically decreasing} in $\Delta$
\item $I_{hd,2}$ is \emph{monotonically increasing} in $\Delta$
\end{itemize}

To proceed, let us first fix an $f$, and find the optimal $\Delta = \Delta^*(f)$ corresponding to that fixed value of $f$. The equivalent representation of the original optimization problem then becomes
\begin{align*}
&R_{\QMF}^*=\max_{f \in (0,1)}\lbp R_{\Delta}(f) \rbp,\\
&\text{where }R_{\Delta}(f) = \max_{\Delta > 0} \lbp \left\lceil \min \lbp I_{hd,1}, I_{hd,2} \rbp \right\rceil^+ \rbp \\
\Delta^*(f) &= \arg_{\Delta > 0} R_{\Delta}(f),\ \text{and}\
\Delta^*(f) = \arg_{\Delta > 0} R_{\Delta}(f).
\end{align*}

For the boundary values, we have
\begin{align*}
&I_{hd,1}(f, \Delta = 0)= f\log \lp 1 + h^2 + g_2^2 \rp + \lp 1 - f \rp \log \lp 1 + g_2^2 \rp\\
&I_{hd,1}(f, \Delta = \infty) = \log \lp 1 + g_2^2 \rp\\
&I_{hd,2}(f, \Delta = 0) = -\infty\\
&I_{hd,2}(f, \Delta = \infty) = \begin{array}{l}f\log \lp 1 + g_2^2 \rp \\+ \lp 1 - f \rp \log \lp 1 + g_1^2 + g_2^2 \rp \end{array}
\end{align*}

The fact that $I_{hd,1}(f,\Delta) > 0$ and $I_{hd,1}(f, \Delta = 0) > I_{hd,2}(f, \Delta = 0)$, coupled with the above mentioned monotonicity properties of $I_{hd,1}$ and $I_{hd,2}$ in the variable $\Delta$ ensure that for a given $f$, $\Delta^*(f) = \arg_{\Delta>0}\lbp I_{hd,1} = I_{hd,2} \rbp$. We can then run a search over $f \in (0,1)$ to maximize $R_{\Delta}(f)$. The maximising schedule $f = f^*$ and the distortion $\Delta^*(f^*)$ is the required optimal.

\subsection{CSIR at the Relay}

To obtain a jointly optimal quantizer and relay schedule in this setting, the metric to minimize is again the probability of outage, conditioned only on $h$. Proceeding similarly as in the full-duplex case, we have,
\begin{align*}
1 - P_{\Out|h}(R; f,\Delta)&= \Pr \lbp \lbp R <  I_{hd,1} \rbp \bigcap \lbp R <  I_{hd,2} \rbp \rbp.
\end{align*}
We note that both $g_1^2$ and $g_2^2$ are to be treated as random variables in the optimization for the CSIR setting. Since $I_{hd,2}$ now involves \emph{two} random variables, the event $\lbp \lbp R <  I_{hd,1} \rbp \bigcap \lbp R <  I_{hd,2} \rbp \rbp$, for a given $(f,\Delta)$ pair can equivalently be represented as $\lbp \lp g_1^2, g_2^2 \rp \in \mathcal{G}(f,\Delta) \rbp$ where $\mathcal{G}$ is an appropriate open region in the first quadrant of the $2$-D plane. The probability is computed numerically and the computation is sped up by taking advantage of the structure of the integration region (isolating rectangular components) and making use of the cumulative distribution function of exponential random variables. The optimal $(f,\Delta)$ pair is the one that minimizes $P_{\Out|h}(R; f,\Delta)$ thus obtained. One can show that for a given $f$, $P_{\Out|h}(R;f,\Delta)$ increases with $\Delta$ after a certain threshold $\Delta_{th}$. Also, from extensive numerical evaluations, we can conjecture (similar to that in the full-duplex CSIR case for asymmetric fades) that there is exactly one critical point of $P_{\Out|h}(R; f,\Delta)$ in $(0,\Delta_{th})$ which we can efficiently compute numerically. The optimal over $f$ can then be found by a sufficiently quantized search over the $(0,1)$ interval.

\subsection{CSIR-limited hybrid QMF/DDF}
\label{subsec:DDF/QMF}

Similar to the full-duplex case, we outline a hybrid strategy for the CSIR-limited half-duplex network that combines the benefits of the CSIR-optimal QMF and the dynamic decode forward (DDF) strategy proposed in \cite{Azarian_05}. The relay operation here is more involved than that of the full-duplex relay, as it has to jointly choose a schedule, a quantizer and a relaying scheme (DDF or QMF) using the information available. We next outline the relay operation for this scheme.

Given the received channel $\msf{h}$, the relay essentially runs a scan over all (sufficiently finely quantized) schedules $f \in [0,1]$, and for each such $f$, computes the conditional probability of outage. If a schedule $f$ permits a decoding operation (i.e., $R < f\log(1 + h^2)$), the outage probability corresponds to that of DDF, whereas, if it does not permit decoding, the outage probability corresponds to that of CSIR-optimized QMF. The schedule $f^*$ that minimizes the conditional outage probability thus obtained is the required schedule, and the relaying scheme to be used is the one that leads to this minimum outage probability.

However, one can observe that the properties of the rate expressions for DDF allow for a lower complexity equivalent operation at the relay. Consider the quantity $f_{\DDF} = \frac{R}{\log (1 + h^2)}$. If, for the given $h$, $f_{\DDF}$ is inadmissible (i.e., $f_{\DDF} > 1$), the relay will use CSIR-optimized QMF. For the case when $f_{\DDF}$ is admissible, the values of $f$ (in the scan at the relay) for which DDF will be considered to compute the probability of outage lie in the interval $[f_{\DDF}, 1]$. Observe that in this regime, the probability of outage is an increasing function of $f$, and hence one only needs to run the scan in the interval $[0, f_{\DDF}]$ with the outage probability corresponding to QMF for $f \in [0,f_{\DDF})$ and to DDF for $f = f_{\DDF}$.

We next prove that the outage performance of the hybrid scheme described above is superior to that of CSIR-optimized QMF as well as DDF schemes taken individually. To do so, we demonstrate that for every schedule $f \in [0,1]$, the outage probability of the hybrid scheme is smaller than that of DDF and CSIR-optimized QMF, which naturally translates to it being better when the optimizing $f$ is chosen. We denote by $P_{\mathsf{out}}^{(f)}$, the probability of outage at a particular schedule $f$. Writing out these probabilities for the schemes, we have:

{\small
\begin{align}
P_{\Out,\DDF}^{(f)}(R) &= \begin{array}{l}\lbp\begin{array}{l}\Pr \lbp R > \begin{array}{c}(1 - f) \log \lp 1 + g_1^2 + g_2^2 \rp \\ + f \log \lp 1 + g_2^2 \rp \end{array} \rbp \\. \Pr \lbp R < f \log \lp 1 + h^2 \rp \rbp\end{array}\rbp \\
+ \lbp\begin{array}{l}\Pr \lbp R > \log \lp 1 + g_2^2 \rp \rbp \\. \Pr \lbp R \geq f \log \lp 1 + h^2 \rp \rbp \end{array}\rbp\end{array} \label{eq:P_out_DDF}\\
P_{\Out,\QMF}^{(f)}(R) &= \begin{array}{l}\lbp\begin{array}{l}\Pr \lbp R > R_{hd,\QMF}(\Delta^*(f)) \rbp \\. \Pr \lbp R < f \log \lp 1 + h^2 \rp \rbp\end{array}\rbp \\
+ \lbp\begin{array}{l}\Pr \lbp R > R_{hd,\QMF}(\Delta^*(f)) \rbp \\. \Pr \lbp R \geq f \log \lp 1 + h^2 \rp \rbp \end{array}\rbp\end{array} \label{eq:P_out_QMF_hd}\\
P_{\Out,\HYB}^{(f)}(R) &= \begin{array}{l}\lbp\begin{array}{l}\Pr \lbp R > \begin{array}{c}(1 - f) \log \lp 1 + g_1^2 + g_2^2 \rp \\ + f \log \lp 1 + g_2^2 \rp \end{array} \rbp \\. \Pr \lbp R < f \log \lp 1 + h^2 \rp \rbp\end{array}\rbp \\
+ \lbp\begin{array}{l}\Pr \lbp R > R_{hd,\QMF}(\Delta^*(f)) \rbp \\. \Pr \lbp R \geq f \log \lp 1 + h^2 \rp \rbp \end{array}\rbp\end{array} \label{eq:P_out_HYB_hd}
\end{align}}

To show that $P_{\Out,\HYB}^{(f)}(R) < P_{\Out,\QMF}^{(f)}(R)$, we observe from equations. (\ref{eq:P_out_QMF_hd}) and (\ref{eq:P_out_HYB_hd}) that we have to show:
\begin{align*}
&\Pr \lbp R > (1 - f) \log \lp 1 + g_1^2 + g_2^2 \rp + f \log \lp 1 + g_2^2 \rp \rbp \\
&< \Pr \lbp R > R_{hd, \QMF}(\Delta^*(f)) \rbp
\end{align*}where $\Delta^*(f)$ is the CSIR-optimal quantizer distortion derived previously.

From the QMF rate expression in \eqref{eq:rate_HD_single}, we see that

\begin{align*}
&R_{hd, \QMF}(\Delta^*(f)) \\
& < \lbp\begin{array}{c}(1 - f) \log \lp 1 + g_1^2 + g_2^2 \rp \\ + f \log \lp 1 + g_2^2 \rp - f \log \lp \frac{1 + \Delta^*(f)}{\Delta^*(f)} \rp \end{array}\rbp\\
\implies & R_{hd, \QMF}(\Delta^*(f)) \\
& < (1 - f) \log \lp 1 + g_1^2 + g_2^2 \rp + f \log \lp 1 + g_2^2 \rp
\end{align*} which proves that $P_{\Out,\HYB}^{(f)}(R) < P_{\Out,\QMF}^{(f)}(R)$.

Next, to show that $P_{\Out,\HYB}^{(f)}(R) < P_{\Out,\DDF}^{(f)}(R)$, we observe from equations (\ref{eq:P_out_DDF}) and (\ref{eq:P_out_HYB_hd}) that we have to show:
\begin{align*}
\Pr \lbp R > R_{hd, \QMF}(\Delta^*(f)) \rbp < \Pr \lbp R > \log \lp 1 + g_2^2 \rp \rbp
\end{align*}where $\Delta^*(f)$ is the CSIR-optimal quantizer distortion derived previously.

From the definition of $\Delta^*(f)$, we know that it minimizes $\Pr \lbp R > R_{hd, \QMF}(\Delta) \rbp$ for every realization $h$ of the $\msf{S}-\msf{R}$ channel at any given schedule. Hence, if we pick any $\Delta = \Delta_1$ irrespective of $h$ and $f$ (in other words, we pick the same quantizer distortion $\Delta_1$ for all realizations $h$ instead of optimized distortions for every realization), we have:
\begin{align*}
\Pr \lbp R > R_{hd, \QMF}(\Delta^*(f)) \rbp < \Pr \lbp R > R_{hd, \QMF}(\Delta_1) \rbp
\end{align*}
Picking $\Delta_1 = \infty $ gives:
\begin{align*}
\Pr \lbp R > R_{hd, \QMF}(\Delta^*(f)) \rbp < \Pr \lbp R > \log \lp 1 + g_2^2 \rp \rbp
\end{align*}which proves that $P_{\Out,\HYB}^{(f)}(R) < P_{\Out,\DDF}^{(f)}(R)$.

\subsection{Numerical Evaluations}

\begin{figure}[!h]
\centering
\subfigure[All channels are i.i.d. The rate scales as: $R = 0.3\log_2(\mathsf{SNR})$]{
\includegraphics[scale=0.4]{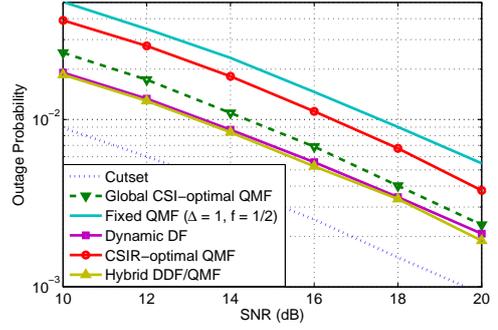}
\label{res:HD_0.3_iid}
}
\subfigure[All channels are i.i.d. The rate scales as: $R = 0.7\log_2(\mathsf{SNR})$]{
\includegraphics[scale=0.4]{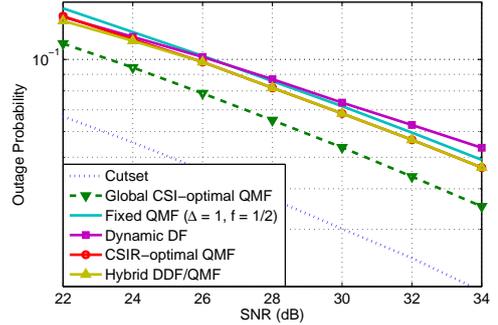}
\label{res:HD_0.7_iid}
}
\caption[]{Outage Performance in Rayleigh faded channels over the half-duplex single-relay network}
\label{res:singlerelay_HD}
\end{figure}
Fig. \ref{res:HD_0.3_iid} and \ref{res:HD_0.7_iid} depict the outage performance of the relaying schemes for two different multiplexing gains--$0.3$ and $0.7$ over i.i.d Rayleigh-faded channels\footnote{The schedule for the cutset bound expression in \eqref{eq:cutset_HD} is optimized using global CSI, similar to Sec.~\ref{subsec:globalCSI_HD} and is given by $f^* = \arg_{f}\lbp C_{hd,1} = C_{hd,2} \rbp = \frac{\log \lp \frac{1 + \lp g_1 + g_2 \rp^2}{1 + g_2^2} \rp}{\log \lp \frac{\lp 1 + h^2 + g_2^2 \rp \lp 1 + \lp g_1 + g_2 \rp^2 \rp}{\lp 1 + g_2^2 \rp^2} \rp}$}. In both cases, we see that using CSIR-optimal QMF provides an improvement of $\sim 1$ dB over the baseline scheme used in \cite{PawarAvestimehr_08} to prove DMT optimality. In addition, the global CSI-optimal provides an additional $3$ dB gain over the CSIR optimal in settings where feedback can be exploited.

For multiplexing gains less than $0.5$, both QMF and DDF are DMT-optimal for this network, but the plots in Fig. \ref{res:HD_0.3_iid} clearly demonstrate the edge that DDF has over even optimized QMF in finite SNRs at these rates. The situation changes around for Fig. \ref{res:HD_0.7_iid} however, as in this regime DDF is not DMT-optimal \cite{PawarAvestimehr_08} and is outperformed by CSIR-optimal QMF.

As proved in Section \ref{subsec:DDF/QMF}, the hybridization of CSIR-optimal QMF and DDF consistently outperforms both the schemes in the two settings shown.

\section{Diamond Networks}
\label{sec:multiple_relays}
In this section, we demonstrate how the performance of QMF can be optimized for the $N$-relay diamond network, and how the topology and limited CSI at the relays can be utilized for parameter optimizations. We start off by showing that an \emph{universal} quantizer distortion level at all relays can be appropriately chosen (without the need for CSI) to bound the worst-case gap of QMF to within $\Theta(\log(N))$ bits/s/Hz for an $N$-relay network. Next, we provide analytical solutions to the (non-convex) quantizer optimization problem for the $2$-relay network and also for the symmetrical $N$-relay network when global CSI is available at the relays.
For the CSIR setting, unfortunately we cannot extend the optimization framework in the single-relay case to optimize over quantization levels, as the relay nodes do not have the same knowledge about the network and the objective function in this optimization problem is not decomposable in general.
Instead, we show how CSIR at the relay can be used to invoke a hybrid strategy involving DF and universal QMF that improves the outage performance of the latter.
\subsection{Universal Quantizers Achieve Capacity within $\Theta(\log(N))$ Bits}

Following the reasoning in \cite{LimKim_10}, we can further obtain a better \emph{universal} quantizer distortion (in the sense of \emph{worst-case} gap over all possible channel realizations in the network) independent of channel coefficients. Setting the quantizer distortions to be the same as $\Delta$, we have the following achievable rate: ($R\lp\Omega; \Delta_{[1:N]}\rp$ is defined in \eqref{eq_QMFRate})

\begin{align*}
&R\lp\Omega; \Delta_{[1:N]} = (\Delta,\ldots,\Delta)\rp\\
&= \lb\begin{array}{l}
\log\lp 1+\sum\limits_{i\in\Omega}g_i^2\rp + \log\lp1+\sum\limits_{j\in\Omega^c}\frac{h_j^2}{1+\Delta}\rp\\
- |\Omega|\log\frac{1+\Delta}{\Delta}\end{array}
\rb^+
\end{align*}
With a term-by-term comparison with the upper bound \eqref{eq_MinCut}, we see that the worst-case gap is upper bounded by
\begin{align}
&\mathrm{gap}^*(\Delta;N) \notag \\
\leq& \max \lbp \begin{array}{l} \log(N) + N\log\lp \frac{1 + \Delta}{\Delta} \rp,\\ \lbp\begin{array}{l} \log(N - i) + (N - i)\log \lp \frac{1 + \Delta}{\Delta} \rp \\ + \log \lp 1 + \Delta \rp \end{array}\rbp_{i \in [1:N]} \end{array}\rbp \notag \\
=& \max \lbp \begin{array}{l} \log(N) + N\log\lp \frac{1 + \Delta}{\Delta} \rp,\\ \log(N - 1) + (N - 1)\log \lp \frac{1 + \Delta}{\Delta} \rp + \log \lp 1 + \Delta \rp \end{array}\rbp \label{eq:gap_universal}
\end{align}

Also, one can observe that the bound on the worst-case gap in \eqref{eq:gap_universal} is indeed achieved for certain configurations of the network. The first term in the maximization is achieved when $\Omega = \phi$ is the mincut in both the cutset and QMF rate expressions. A configuration of the following type makes this possible: Let $g_i^2 = \rho^{\alpha} \quad\forall i, h_i^2 = \rho^{\beta} \quad\forall i$ with $\alpha < \beta$ and $\rho \rightarrow \infty$. The second term can be achieved with a configuration of the following type: Let $h_1^2 = \rho^{\alpha}, h_{i,i \in \{2,\ldots,N\}}^2 = \rho^{\alpha'}, g_1^2 = \rho^{\beta'}, g_{i,i \in \{2,\ldots,N\}}^2 = \rho^{\beta}$ with $\alpha < \alpha'$, $\beta < \beta'$ and $\rho \rightarrow \infty$. Thus, the expression in \eqref{eq:gap_universal} can be interpreted as the worst-case gap for a given universal quantizer distortion $\Delta$, as opposed to simply a upper bound on it.

For the above, we now wish to find an optimal $\Delta = \Delta_{\mathsf{opt}}$ that solves the following problem:
\begin{align*}
\Delta_{\mathsf{opt}} &= \arg \min_{\Delta} \mathrm{gap}^*(\Delta;N)
\end{align*}

We note when $\Delta < \frac{N}{N - 1}$, the first term inside the $\max$ in \eqref{eq:gap_universal} is larger, and when $\Delta \geq \frac{N}{N - 1}$, the second term dominates. Taken separately, we see that the first term is monotonically decreasing in $\Delta$ and the second term has exactly one minima which occurs at $\Delta = N - 1$. We observe that if this minima occurs before (in the order of increasing $\Delta$) $\Delta = \frac{N}{N - 1}$, the optimizing $\Delta$ is $\Delta_{\mathsf{opt}} = \frac{N}{N - 1}$. If however, the minima occurs after $\Delta = \frac{N}{N - 1}$, $\Delta_{\mathsf{opt}} = N - 1$. Combining all these, we have the following:
\begin{align*}
\Delta_{\mathsf{opt}} = \begin{cases} \frac{N}{N - 1}, \quad \frac{N}{N - 1} > N-1 \Rightarrow N < 2.618 \Rightarrow N = 2 \\
N - 1, \quad \frac{N}{N - 1} < N-1 \Rightarrow N > 2.618 \Rightarrow N > 2\end{cases}
\end{align*}

This choice of $\Delta_{\mathsf{opt}}$ leads to the following result on the \emph{best} worst-case gap in an $N$-relay diamond network with universal quantizers:
\begin{align*}
    \mathrm{gap}^*(N) &= \begin{cases}2\log(3) - 1, \quad N = 2 \\
    N\log\lp\frac{N}{N - 1}\rp + 2\log \lp N - 1 \rp, \quad N > 2 \end{cases}
\end{align*}

\begin{figure}[!h]
          \centering
          \includegraphics[width=0.7\columnwidth]{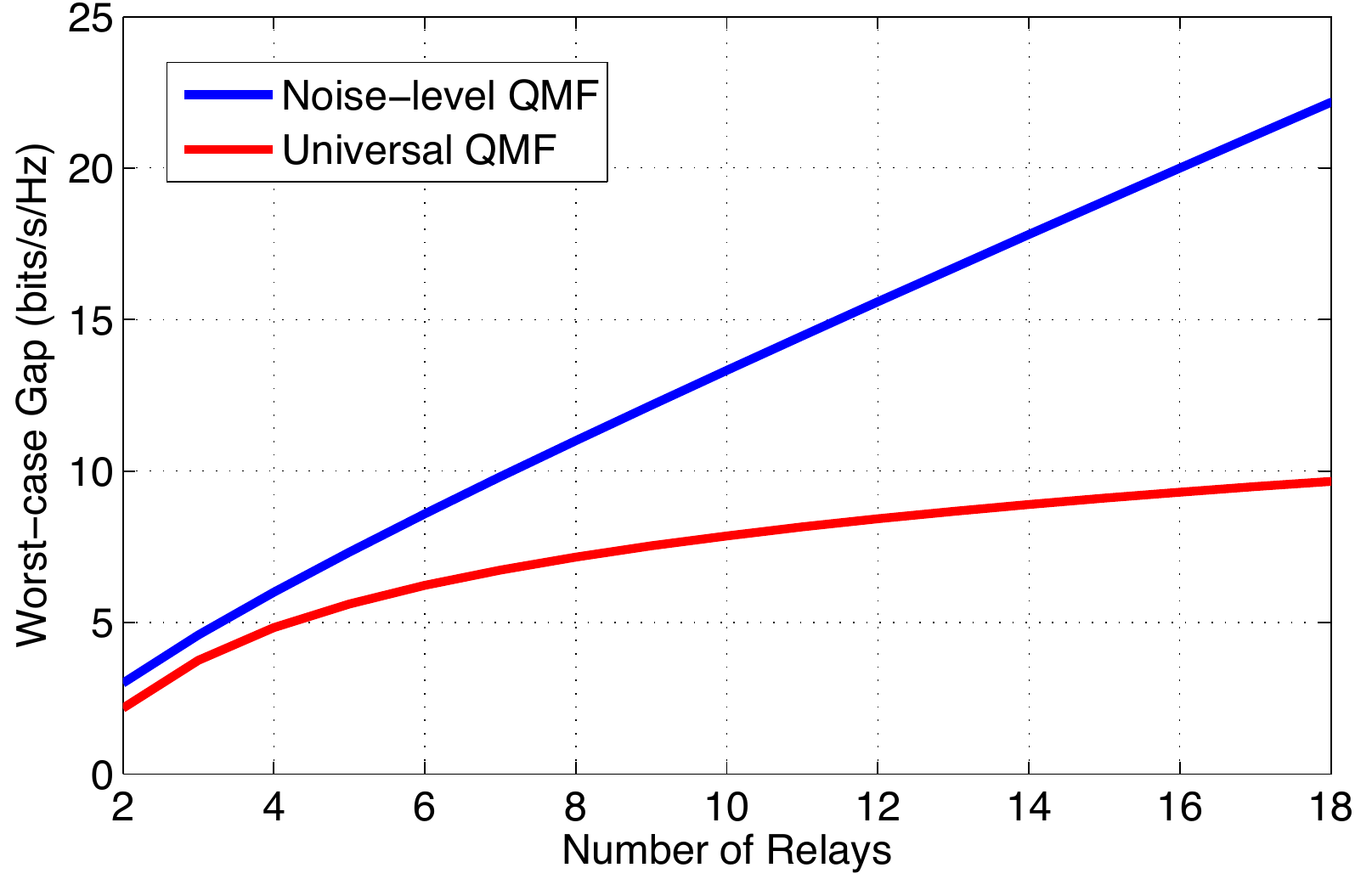}
          \caption{Comparison of worst-case gaps with number of relays}\label{fig:universal_vs_noiselevel}
\end{figure}

One can see that for large $N$, the worst-case gap with optimized universal quantizers grows approximately logarithmically in $N$, as
\begin{align*}
\lim_{N\ra \infty} N\ln\frac{N}{N-1} = \lim_{N\ra\infty} \frac{\frac{1}{N} - \frac{1}{N-1}}{-1/N^2} = \lim_{N\ra\infty} \frac{N}{N-1} = 1.
\end{align*}
In contrast, the worst-case gap with noise-level quantization is dominated by a linear term in $N$. Fig. \ref{fig:universal_vs_noiselevel} makes the benefits of using an optimized universal quantizer derived above explicit, in terms of the gap performance achieved.

\subsection{Channel-Aware Quantizer Optimization}

With global CSI at the relays, the optimization problem can be stated as:
\begin{align}
R^*_{\QMF} &= \max_{\Delta_{[1:N]}\ge 0} R_{\QMF}\lp  \Delta_{[1:N]}\rp \notag \\
&= \max_{\Delta_{[1:N]}\ge 0}\min_{\Omega}\lbp R\lp\Omega; \Delta_{[1:N]}\rp\rbp. \label{eq_NRelay_G_Opt}
\end{align}

We note that this optimization is not convex, as within the minimization part of \eqref{eq_NRelay_G_Opt}, for $\Omega = \emptyset$, the function $R\lp\Omega; \Delta_{[1:N]}\rp$ is not concave in $\Delta_{[1:N]}$. Instead, we provide an analytical characterization of the optimizing distortions and the corresponding achievable rates for
the $2$-relay network, and for symmetric $N$-Relay networks.

\subsubsection*{Solution for the $2$-Relay Network}

In order to solve the optimization problem in \eqref{eq_NRelay_G_Opt} for the case $N=2$,
we consider the following equivalent formulation:
\begin{align}\notag
R^*_{\text{QMF,G}} &= \max_{\Delta_2 > 0} \lbp \max_{\Delta_1 > 0}\min_{\Omega} R\lp\Omega;\Delta_1,\Delta_2\rp \rbp\\
&:= \max_{\Delta_2 > 0} R_{\text{QMF,G}}^*(\Delta_2) \label{eq_2Relay_G_Opt}
\end{align}
with $R_{\text{QMF,G}}^*(\Delta_2) = \max_{\Delta_1 > 0}\min_{\Omega} R\lp\Omega;\Delta_1,\Delta_2\rp $, $\Delta_2 >0$.

We shall first characterize $R_{\text{QMF,G}}^*(\Delta_2)$ and then optimize it over $\Delta_2$ to obtain the solution. The following lemma and theorem summarize the main result.

\begin{lemma}[Characterization of {\rm $R_{\text{QMF,G}}^*(\Delta_2)$}] \label{lem_2Relay}
Let us define the following intervals of $\Delta_2$:
\begin{align*}
\mcal{I}_1 &:= \lp 0, \delta_1\rp & \mcal{I}_2 &:=  \lb \delta_1, \delta_2\rp &
\mcal{I}_3 &:= \lb \delta_2,\infty\rp
\end{align*}
where
$
\delta_1 := \frac{(1+g_1^2+g_2^2)(1+h_1^2+h_2^2) + (1+g_2^2)h_1^2h_2^2}{g_2^2(1+g_1^2+g_2^2)(1+h_1^2)}$, and
$
\delta_2 := \frac{(1+g_1^2)(1+h_2^2)}{g_2^2}$.
In each range of $\Delta_2$, the optimizing $\Delta_1^*$ and $\Omega^*$ in the max-min problem $\max_{\Delta_1 > 0}\min_{\Omega} R\lp\Omega;\Delta_1,\Delta_2\rp$ is given as follows:
\begin{align*}
\bullet\quad &\Delta_2\in\mcal{I}_1: & &\\
&\Delta_1^* = \frac{(1+g_2^2)(1+h_1^2)}{g_1^2} &
&\Omega^* = \{1,2\}\mbox{ or } \{2\}\\
\bullet\quad &\Delta_2\in\mcal{I}_2: & &\\
&\Delta_1^* = \frac{(1+h_1^2)\Delta_2+(1+h_1^2+h_2^2)}{(g_1^2+g_2^2)\Delta_2-(1+h_2^2)} &
&\Omega^* = \{1,2\}\mbox{ or } \emptyset\\
\bullet\quad &\Delta_2\in\mcal{I}_3: & &\\
&\Delta_1^*  = \frac{(1+h_1^2)\Delta_2 + (1+h_1^2+h_2^2)}{g_1^2(\Delta_2+(1+h_2^2))} &
&\Omega^* = \{1\}\mbox{ or } \emptyset
\end{align*}
Moreover we always have $0< \delta_1 < \delta_2$, and hence the three intervals $\mcal{I}_1,\mcal{I}_2,\mcal{I}_3$ are not empty.
\end{lemma}
\begin{proof}
See Appendix \ref{app_Pf_lem_2Relay}
\end{proof}
\begin{theorem}\label{thm_2Relay}
Let
$
A := h_1^2(1+h_1^2) - h_2^2(1+h_1^2+g_1^2+g_2^2)$,
$B := 2h_1^2(1+h_1^2)$,
$C := h_1^2(1+h_1^2+h_2^2)$,
and
$
\delta_3 := \frac{-B-\sqrt{B^2-4AC}}{2A}
$.

%
The solution to the maximization problem in \eqref{eq_2Relay_G_Opt}, is summarized below:
{\flushleft 1) $A\ge 0$ or $\delta_3 \in \mcal{I}_3$}: The optimal solution is
\begin{align*}
\Delta_2^* &= \delta_2 &
\Delta_1^* &= \frac{(1+h_1^2)\delta_2+(1+h_1^2+h_2^2)}{(g_1^2+g_2^2)\delta_2-(1+h_2^2)}\\
& & &= \frac{(1+h_1^2)\delta_2 + (1+h_1^2+h_2^2)}{g_1^2(\delta_2+(1+h_2^2))}
\end{align*}
{\flushleft 2) $A < 0$ and $\delta_3 \in \mcal{I}_1$}: The optimal solution is
\begin{align*}
\Delta_2^* &= \delta_1&
\Delta_1^* &= \frac{(1+h_1^2)\delta_1+(1+h_1^2+h_2^2)}{(g_1^2+g_2^2)\delta_1-(1+h_2^2)}\\
& & &= \frac{(1+g_2^2)(1+h_1^2)}{g_1^2}
\end{align*}
{\flushleft 3) $A < 0$ and $\delta_3 \in\mcal{I}_2$}: The optimal solution is
\begin{align*}
\Delta_2^* &= \delta_3 &
\Delta_1^* &= \frac{(1+h_1^2)\delta_3+(1+h_1^2+h_2^2)}{(g_1^2+g_2^2)\delta_3-(1+h_2^2)}
\end{align*}
\end{theorem}
\begin{proof}
See Appendix \ref{app_Pf_thm_2Relay}
\end{proof}
\subsubsection*{Solution for the Symmetric $N$-Relay Network}
We consider the case where $h_i = h$ and $g_i = g$ for all $i=1,\ldots,N$. By symmetry, the optimal distortion level $\Delta_i = \Delta$ for all $i=1,\ldots,N$, and the optimization becomes
\begin{align*}
&R^*_{\QMF} = \max_{\Delta\ge0}\min_{k\in[0:N]} R_k(\Delta),
\end{align*}
where $R_k(\Delta) := \log\lp 1+(N-k)\frac{h^2}{1+\Delta}\rp + \log\lp 1+kg^2\rp - k\log\lp\frac{1+\Delta}{\Delta}\rp$.
By plotting the $R_k(\Delta)$'s for $k\in[0:N]$ as a function of $\Delta$ for various $N, h \mbox{ and } g$ combinations, one can observe that the $\max\min$ optimum appears to occur at the value of $\Delta$ where the curves $R_0(\Delta)$ and $R_N(\Delta)$ intersect. A demonstrative plot is shown for a $5$-relay network in Fig. \ref{proofdemo_NRelay}. This observation can in fact be proved, for which the following lemmas are necessary (proofs in Appendix \ref{app_Pf_lem1_NRelay}):
\begin{figure}[!h]
  \centering
  \includegraphics[width=0.7\columnwidth]{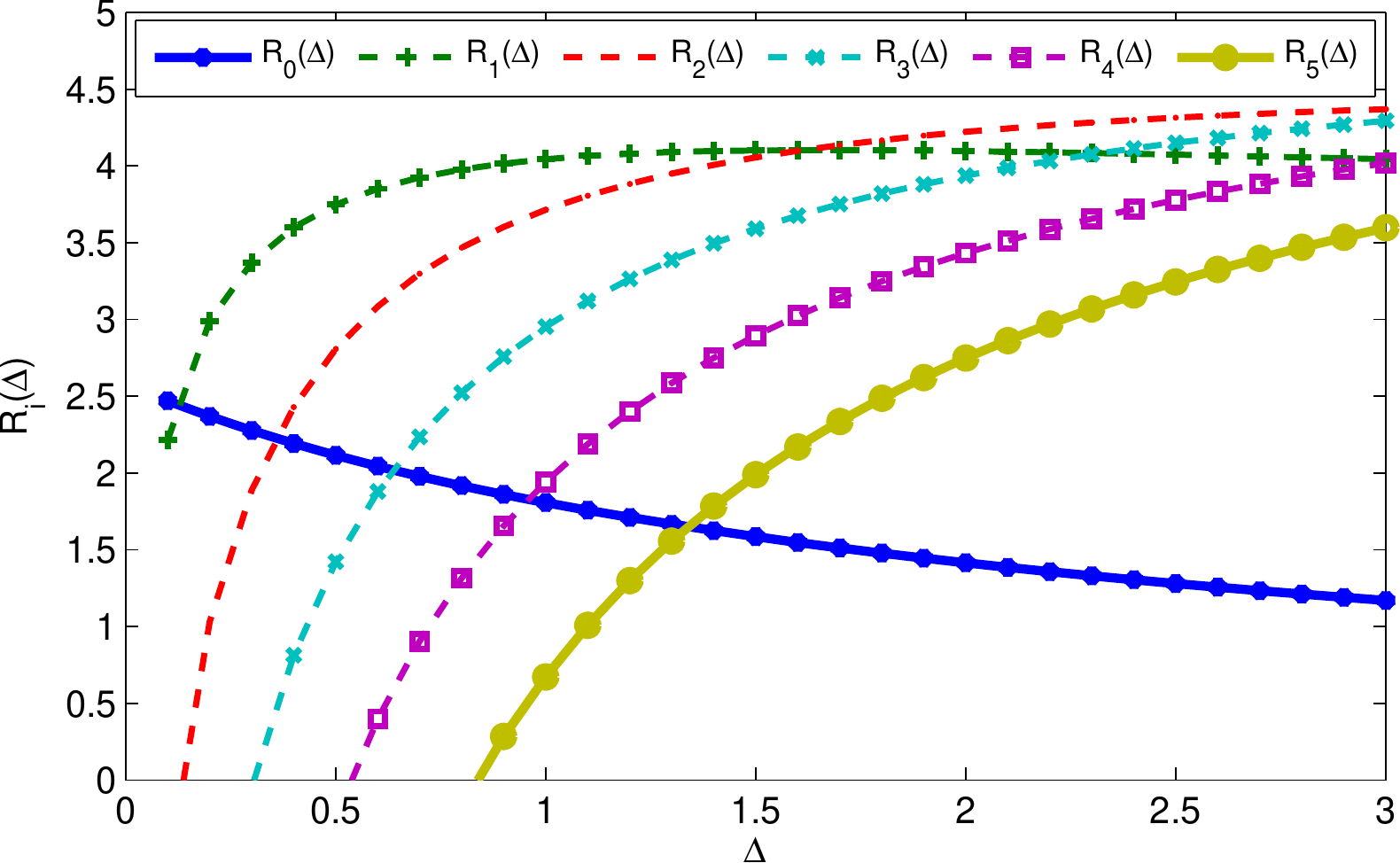}\\
  \caption{$R_i(\Delta)$'s as a function of $\Delta$ for a $5$-Relay symmetric network with $|h| = 1$ and $|g| = 3.16$.}\label{proofdemo_NRelay}
\end{figure}
\begin{lemma} \label{lem1_NRelay}
There exists exactly one positive $\Delta$ = $\Delta^{*}_{ij}$ satisfying $R_i(\Delta) = R_j(\Delta) \quad \forall i,j \in [0:N]$ and $i\neq j$.
\end{lemma}
\begin{lemma} \label{lem2_NRelay}
$\lim_{\Delta\rightarrow 0} \{R_i(\Delta) - R_{i+1}(\Delta)\} \geq 0 \quad \forall i \in [0:N-1]$.
\end{lemma}
\begin{lemma} \label{lem3_NRelay}
Let $\Delta^{*}_{(i)(i+1)}$ be the unique positive solution of $R_i(\Delta) = R_{i+1}(\Delta) \quad \forall i \in [0:N-1]$. Then, $\Delta^{*}_{(i)(i+1)}$ is non-decreasing in $i$, i.e., $\Delta^{*}_{01}\leq \Delta^{*}_{12}\leq \ldots \leq \Delta^{*}_{(N-1)(N)}$.
\end{lemma}
These lemmas ensure that each $R_k(\Delta)$, $k \in [1:N]$ intersects $R_0(\Delta)$ \emph{before} (in the order of increasing $\Delta$) any other curve. Since $R_0(\Delta)$ is a \emph{decreasing} function of $\Delta$, the initial observation is validated and can be crystalized as follows:

\begin{theorem}\label{thm_NRelay_sym}
$\left(R_0(\Delta^{*}_{0N}),\Delta^{*}_{0N}\right)$ attains optimum, where $\Delta^{*}_{0N}$ is the root of $R_0(\Delta) = R_N(\Delta)$.
\end{theorem}

\subsection{CSIR limited Hybrid DF/QMF Relaying}

\label{subsec:hybrid_diamond}

In the spirit of the hybrid scheme that we proposed for the single relay network, a CSIR-limited hybrid scheme for the diamond network is illustrated in this subsection. In this scheme, the relays individually decide whether or not they can decode the source message, depending on the incoming source-to-relay channel $\msf{h}_i$ for the $i^{th}$ relay. For the relays where the incoming channel supports decoding, the relays decode and transmit the re-encoded message. The relays that cannot decode apply QMF, i.e., they quantize, map and forward their received signals to the destination. Each relay sends a $1$-bit flag (in a possibly orthogonalized preamble, so that these do not interfere) to let the destination know whether it used DF or QMF mode of operation.

In the following, we will demonstrate that this strategy, strictly improves upon the outage performance of QMF for all possible choices of quantizer distortions. In fact, as we illustrate in the next section, simulations even demonstrate significant benefits of this hybrid scheme over the DF scheme alone, which is known to not achieve optimal DMT or constant-gap performance for all channel configurations of the diamond network.

Writing out the probability of outage for the hybrid scheme, we have:
\begin{align*}
P_\Out(R) &= \sum_{\Omega} \lbp\begin{array}{l}\Pr \lbp R > R_{\HYB} \lp \Omega_D \rp | \Omega_D = \Omega \rbp\\ . \Pr \lbp \Omega_D = \Omega \rbp \end{array}\rbp
\end{align*}where $\Omega_D$ denotes the subset of nodes that can decode the source message, and $\Omega$ denotes an arbitrary partition on the set of relays.

It is important to note that the supportable rate for the hybrid scheme, $ R_{\HYB} \lp \Omega_D \rp $ is dependent on the set of nodes that can decode. For instance, in the term corresponding to $\Omega_D = \phi$, i.e., no relay can decode, $ R_{\HYB} \lp \Omega_D = \phi \rp = R_{\QMF}$.

One can observe that the events $\lbp R > R_{\HYB} \lp \Omega_D \rp | \Omega_D = \Omega \rbp$ and $\lbp R > I_{\mathsf{D}} \lp \Omega_D = \Omega \rp \rbp $ are equivalent, where $I_{\mathsf{D}} \lp \Omega_D = \Omega \rp$ is given by the following expression (see Fig. \ref{fig:illustration_cut_hybrid_diamond} for illustration):
\begin{figure}[!h]
  \centering
  \includegraphics[width=0.75\columnwidth]{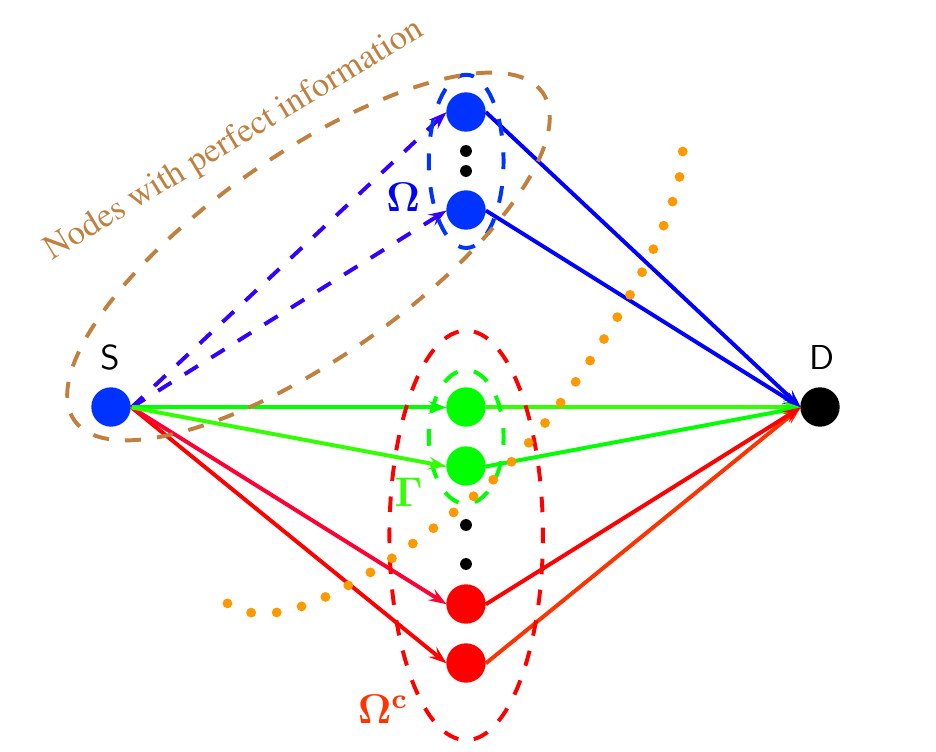}\\
  \caption{The cut in the above diagram depicts the information transfer across nodes for a given choice of $\Gamma$. $I_D(\Omega)$ is the minimum of all such cut values computed over every possible choice of $\Gamma$}\label{fig:illustration_cut_hybrid_diamond}
\end{figure}
\begin{align}
&I_{\mathsf{D}} \lp \Omega_D = \Omega \rp \notag \\
=& \min_{\Gamma}\lbp \begin{array}{l} \log \lp 1 + \sum_{i \in \lp \Omega \cup \Gamma \rp} g_i^2 \rp + \\ \log \lp 1 + \sum_{j \in \lp \Omega \cup \Gamma \rp^c} \frac{h_j^2}{1 + \Delta_j} \rp - \sum\limits_{k \in \Gamma} \log \lp \frac{1 + \Delta_k}{\Delta_k} \rp \end{array} \rbp \label{eq:rate_hybrid_diamond}
\end{align}where $\Gamma$ denotes a partition on the set $\Omega^c$.

We observe that to show HYB outperforms QMF, it is sufficient to show that $I_{\mathsf{D}} \lp \Omega \rp > R_{\QMF} \quad \forall \quad \Omega$. This is because the QMF rate expression does not depend on the choice of $\Omega_D$, the set of nodes that can decode the message.

To show this, we note that for every $\Omega$ and every $\Gamma$ (for a given $\Omega$) in eq. (\ref{eq:rate_hybrid_diamond}), we can find a partition $\Omega_{\QMF} = \Omega \cup \Gamma$ for which the QMF mutual information, $I_{\mathsf{Q}}(\Omega_{\QMF})$, is less than the corresponding term inside the minimization expression for $I_{\mathsf{D}} \lp \Omega \rp$ by an amount $\sum_{\omega \in \Omega} \log \lp \frac{1 + \Delta_{\omega}}{\Delta_{\omega}} \rp$. Also, since $R_{\QMF} = \min_{\Omega_{\QMF}}{I_{\mathsf{Q}}(\Omega_{\QMF})}$, we can conclude that $I_{\mathsf{D}} \lp \Omega \rp > R_{\QMF} \quad \forall \quad \Omega$, which proves the result.

\subsection{Numerical Evaluations}

\begin{figure}[!h]
\centering
\subfigure[All channels are i.i.d. The rate scales as: $R = 0.7\log_2(\mathsf{SNR})$]{
\includegraphics[scale=0.4]{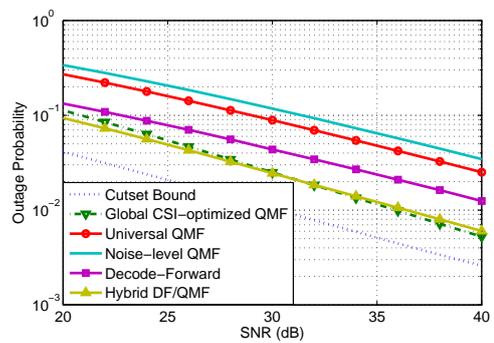}
\label{res:FD_diamond_iid}
}
\subfigure[$\mathbb{E}\{|h_1|^2\} = \mathbb{E}\{|g_2|^2\} = \frac{\mathbb{E}\{|g_1^2|\}}{25} = \frac{\mathbb{E}\{|h_2^2|\}}{10}$; $x$-axis: $\mathbb{E}\{|h_1|^2\}$, $R = 0.7\log_2(\mathsf{SNR})$
]{
\includegraphics[scale=0.4]{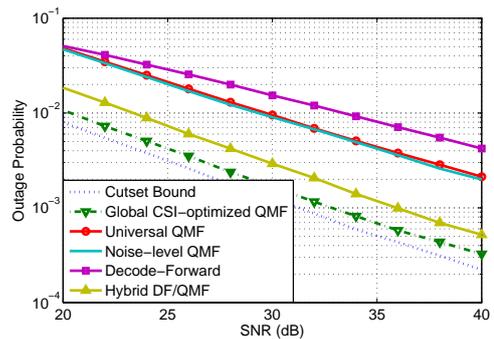}
\label{res:FD_diamond_asym}
}
\caption[]{Outage Performance in Rayleigh faded channels over the $2$-relay diamond network}
\label{res:diamond_FD}
\end{figure}

Fig. \ref{res:diamond_FD} shows the outage performance of the relaying schemes we consider for the diamond network under different channel configurations. Fig. \ref{res:FD_diamond_iid} depicts the performance for i.i.d channels, whereas Fig. \ref{res:FD_diamond_asym} looks at an antisymmetric network configuration. As we can see, in both the above settings, if global CSI of the channel strengths is available to the relays, QMF consistently provides very good performance. In the absence of global CSI however, we see that for i.i.d channels, DF provides a significant advantage over universal QMF, and the noise-level QMF is even worse. This is as per expectations, as the benefits of QMF are more pronounced in asymmetric settings, i.e., when one of the relays is closer to the source while the other is closer to the destination. Indeed, in Fig. \ref{res:FD_diamond_asym}, where we have asymmetric channel configurations, we see that universal QMF outperforms DF in the diversity order (slope of the curve) itself.

It was shown in section \ref{subsec:hybrid_diamond} that the hybrid scheme can provably outperform the universal QMF scheme for the diamond network. What is even more interesting to note from the plots is that not only does it outperform universal QMF, it also significantly outperforms DF in both i.i.d and asymmetric network configurations. It strikes a balance between extracting the multiple-input single-output (MISO) gains of the DF scheme and the optimal DMT performance of QMF to achieve excellent finite SNR performance, even rivaling that of global-CSI optimal QMF.

\appendices

\section{Optimal Quantizer for Asymmetric Fades}\label{app_asym}
When $\lambda_1 \neq \lambda_2$, we have,
\begin{align*}
Q'(\Delta)&= \frac{1}{\lp \lambda_2 - \lambda_1 \rp} \lbp \begin{array}{l}\lp \lambda_2 e^{- \lp \lambda_1 \alpha_2 + \lp \lambda_2 - \lambda_1 \rp \left\lceil \alpha_1 \right\rceil^+ \rp} \rp. \\ \lp \lp \lambda_1 - \lambda_2 \rp {\left\lceil \alpha_1 \right\rceil^+}' - \lambda_1 \alpha_2' \rp \\+ \lambda_1 \lambda_2 \alpha_2' e^{-\lambda_2 \alpha_2} \end{array}\rbp
\end{align*}

\begin{table}
  \centering
  \begin{tabular}{c c c}
  Functions & $\lim\limits_{\Delta \rightarrow 0^+}$ & $\lim\limits_{\Delta \rightarrow \infty}$ \\ \hline
  $\alpha_1 = 2^R - \frac{h^2}{1 + \Delta} - 1$ & $2^R - h^2 - 1$ & $2^R - 1$ \\
  $\alpha_2 = 2^R \lp 1 + \frac{1}{\Delta} \rp - 1$ & $\infty$ & $2^R - 1$ \\
  $\alpha_1' = \frac{h^2}{\lp 1 + \Delta \rp^2}$ & $h^2$ & $0^+$ \\
  $\alpha_2' = -\frac{2^R}{\Delta^2}$ & $-\infty$ & $0^-$
  \end{tabular}
  \caption{Limiting values of $\alpha_1, \alpha_2, \alpha_1'$ and $\alpha_2'$}\label{table:limiting values}
\end{table}

Similar to the case $\lambda_1 = \lambda_2$, we start by assuming $\alpha_1 > 0$. From the boundary values in Table \ref{table:limiting values}, we note that $Q(0) = 0$ and $Q(\infty) = e^{-\lambda_2\lp 2^R - 1 \rp}$. Moreover, we observe that $Q'(\infty)<0$, which implies that the function is decreasing just before attaining its limiting value. These conditions dictate that there must be \emph{at least} one finite local maxima of $Q(\Delta)$, and the $\Delta = \Delta^{\S}$ that maximises $Q(\Delta)$ corresponds to the global maxima. The critical points can be found efficiently using numerical methods to find the roots of $Q'(\Delta) = 0$. From plotting $Q(\Delta)$ for various parameter sets $\lp h, R, \lambda_1, \lambda_2 \rp$, we conjecture that there is \emph{exactly} one critical point of $Q(\Delta)$, but we cannot at this point prove the claim due to the general transcendental nature of the expression.

As before, if $\Delta^{\S} > \Delta_t = \frac{h^2}{2^R - 1} - 1$, our assumption on $\alpha_1$ is validated, and $\Delta^* = \Delta^{\S}$. Otherwise, we set $\left\lceil \alpha_1 \right\rceil^+ = 0$ in the $Q(\Delta)$ expression. We also note that in this case, the maximizing $\Delta^*$ will lie in the interval $\lp 0, \Delta_t \rb $. In such a case, $Q'(\Delta) = \frac{\lambda_1\lambda_2\alpha_2'}{\lambda_2 - \lambda_1}\lbp e^{-\lambda_2 \alpha_2} - e^{-\lambda_1 \alpha_2} \rbp > 0 \quad \forall \Delta \in \lp 0,\Delta_t \rb $, i.e., $Q(\Delta)$ is monotonically increasing in $\Delta$. Hence, the maximizing $\Delta^* = \Delta_t$.

Thus, similar to the case where $\lambda_1 = \lambda_2$, we have,
\begin{align*}
\Delta^* &= \max\lbp \Delta^{\S}, \Delta_t \rbp \quad \mbox{for} \quad \lambda_1 \neq \lambda_2
\end{align*}

\section{Proof of Lemma~\ref{lem_2Relay}}\label{app_Pf_lem_2Relay}

Throughout this proof we work under a fixed $\Delta_2$. Hence for notational convenience, let us drop the $\Delta_2$ argument in $R\lp\Omega;\Delta_1,\Delta_2\rp$.

First observe that
\begin{align*}
R\lp \emptyset;\Delta_1\rp &:= \log\lp 1 + \frac{h_1^2}{1+\Delta_1} + \frac{h_2^2}{1+\Delta_2}\rp\\
R\lp 2;\Delta_1\rp &:= \lb\begin{array}{l} \log\lp1+g_2^2\rp + \log\lp 1 + \frac{h_1^2}{1+\Delta_1}\rp\\ - \log\lp\frac{1+\Delta_2}{\Delta_2}\rp\end{array}\rb^+
\end{align*}
are decreasing functions of $\Delta_1$, while
\begin{align*}
R\lp 1;\Delta_1\rp &:= \lb\begin{array}{l} \log\lp1+g_1^2\rp + \log\lp 1 + \frac{h_2^2}{1+\Delta_2}\rp\\ - \log\lp\frac{1+\Delta_1}{\Delta_1}\rp\end{array}\rb^+\\
R\lp 1,2;\Delta_1\rp &:= \lb\begin{array}{l} \log\lp1+g_1^2+g_2^2\rp  - \log\lp\frac{1+\Delta_1}{\Delta_1}\rp\\ - \log\lp\frac{1+\Delta_2}{\Delta_2}\rp\end{array}\rb^+
\end{align*}
are increasing functions of $\Delta_1$. Besides, within the two curves $R\lp 1;\Delta_1\rp$ and $R\lp 1,2;\Delta_1\rp$, one will be no less than the other throughout all positive $\Delta_1$. Therefore to simplify the problem, let us first find the condition for one of them to be the dominant one, and for the remaining one can focus on the relations of the other two decreasing functions with this dominant increasing curve.

\begin{align*}
&R\lp1;\Delta_1\rp \le R\lp 1,2;\Delta_1\rp\\
\iff &
\log\lp1+g_1^2\rp + \log\lp 1 + \frac{h_2^2}{1+\Delta_2}\rp - \log\lp\frac{1+\Delta_1}{\Delta_1}\rp\\
&\le
\log\lp1+g_1^2+g_2^2\rp  - \log\lp\frac{1+\Delta_1}{\Delta_1}\rp - \log\lp\frac{1+\Delta_2}{\Delta_2}\rp\\
\iff &(1+g_1^2)(1+h_2^2+\Delta_2) \le (1+g_1^2+g_2^2)\Delta_2\\
\iff &\Delta_2 \ge \frac{(1+g_1^2)(1+h_2^2)}{g_2^2} = \delta_2
\end{align*}

Below we discuss in two difference cases based on the above condition.

\subsection{$\Delta_2 \ge \delta_2$}
In this case we have
\begin{align*}
&R_{\text{QMF,G}}^*(\Delta_2) = \max_{\Delta_1 > 0}\min\lbp R\lp 1;\Delta_1\rp , R\lp \emptyset;\Delta_1\rp,R\lp 2;\Delta_1\rp\rbp
\end{align*}

Note that at the two extreme values of $\Delta_1$,
\begin{align*}
&R\lp 1;\Delta_1=0\rp = 0 \\
&R\lp 1;\Delta_1=\infty\rp = \log\lp1+g_1^2\rp + \log\lp 1 + \frac{h_2^2}{1+\Delta_2}\rp\\
&R\lp \emptyset;\Delta_1=0\rp = \log\lp 1 + h_1^2 + \frac{h_2^2}{1+\Delta_2}\rp\\
&R\lp \emptyset;\Delta_1=\infty\rp = \log\lp 1 + \frac{h_2^2}{1+\Delta_2}\rp\\
&R\lp 2;\Delta_1=0\rp =  \lb\begin{array}{l} \log\lp1+g_2^2\rp + \log\lp 1 + h_1^2\rp\\ - \log\lp\frac{1+\Delta_2}{\Delta_2}\rp\end{array}\rb^+\\
&R\lp 2;\Delta_1=\infty\rp =  \lb \log\lp1+g_2^2\rp - \log\lp\frac{1+\Delta_2}{\Delta_2}\rp\rb^+
\end{align*}

We conclude that $R\lp 1;\Delta_1\rp$ and $R\lp \emptyset;\Delta_1\rp$ will always intersect at some positive $\Delta_1$, and using the monotonicity of these curves, the intersection occurs at only one point. On the other hand $R\lp 1;\Delta_1\rp$ and $R\lp 2;\Delta_1\rp$ may not. Moreover, using the monotonicity of these curves, we have
\begin{align*}
&\forall\Delta_1>0,\ R\lp 1;\Delta_1\rp \le R\lp 2;\Delta_1\rp\\
\iff &R\lp 1;\Delta_1=\infty\rp \le R\lp 2;\Delta_1=\infty\rp\\
\iff & \log\lp1+g_1^2\rp + \log\lp 1 + \frac{h_2^2}{1+\Delta_2}\rp\\
&\quad\le \log\lp1+g_2^2\rp - \log\lp\frac{1+\Delta_2}{\Delta_2}\rp\\
\iff &\Delta_2 \ge \frac{(1+g_1^2)(1+h_2^2)}{g_2^2-g_1^2}\mbox{ and } g_2^2 > g_1^2
\end{align*}

Hence, if $g_2^2 > g_1^2$ and $\Delta_2 \ge \frac{(1+g_1^2)(1+h_2^2)}{g_2^2-g_1^2}$, then $\Delta_1^* = \Delta_1^*(1;\emptyset):=$ the intersection of $R\lp 1;\Delta_1\rp$ and $R\lp \emptyset;\Delta_1\rp$:
\begin{align*}
\Delta_1^*(1;\emptyset) := \frac{(1+h_1^2)\Delta_2 + (1+h_1^2+h_2^2)}{g_1^2(\Delta_2+(1+h_2^2))}
\end{align*}

Otherwise, we use the monotonicity of these curves to arrive at
$
\Delta_1^* = \min\lbp \Delta_1^*(1;\emptyset), \Delta_1^*(1;2)\rbp
$,
where $\Delta_1^*(1;2):=$ the intersection of $R\lp 1;\Delta_1\rp$ and $R\lp 2;\Delta_1\rp$:
\begin{align*}
\Delta_1^*(1;2) := \frac{(1+g_1^2)(1+h_1^2)\Delta_2}{(g_1^2-g_2^2)\Delta_2+(1+g_1^2)(1+h_2^2)}
\end{align*}

Next we introduce the following claim, the proof of which is at the end of this section.
\begin{claim}\label{claim_Dominance}
Within the range $\Delta_2 \ge \delta_2$, we always have $\Delta_1^*(1;2) > \Delta_1^*(1;\emptyset)$.
\end{claim}

Therefore, $\Delta_1^* = \Delta_1^*(1;\emptyset)$ is the optimal solution for $\Delta_2 \ge \delta_2$, and $\Omega^* = \{1\}$ or $\Omega^* = \emptyset$.

\subsection{$\Delta_2 < \delta_2$}
In this case we have 
\begin{align*}
R_{\text{QMF,G}}^*(\Delta_2) = \max_{\Delta_1 > 0}\min\lbp R\lp 1,2;\Delta_1\rp , R\lp \emptyset;\Delta_1\rp,R\lp 2;\Delta_1\rp\rbp
\end{align*}

Again, note that at the two extreme values of $\Delta_1$,
\begin{align*}
&R\lp 1,2;\Delta_1=0\rp = 0 \\
&R\lp 1,2;\Delta_1=\infty\rp = \lb\log\lp1+g_1^2+g_2^2\rp- \log\lp\frac{1+\Delta_2}{\Delta_2}\rp\rb^+\\
&R\lp \emptyset;\Delta_1=0\rp = \log\lp 1 + h_1^2 + \frac{h_2^2}{1+\Delta_2}\rp\\
&R\lp \emptyset;\Delta_1=\infty\rp = \log\lp 1 + \frac{h_2^2}{1+\Delta_2}\rp\\
&R\lp 2;\Delta_1=0\rp =  \lb\begin{array}{l} \log\lp1+g_2^2\rp + \log\lp 1 + h_1^2\rp\\ - \log\lp\frac{1+\Delta_2}{\Delta_2}\rp\end{array}\rb^+\\
&R\lp 2;\Delta_1=\infty\rp =  \lb \log\lp1+g_2^2\rp - \log\lp\frac{1+\Delta_2}{\Delta_2}\rp\rb^+
\end{align*}

We conclude that $R\lp 1,2;\Delta_1\rp$ and $R\lp 2;\Delta_1\rp$ will always intersect at some positive $\Delta_1$, and using the monotonicity of these curves, the intersection occurs at only one point. On the other hand $R\lp 1,2;\Delta_1\rp$ and $R\lp \emptyset;\Delta_1\rp$ may not. Moreover, using the monotonicity of these curves, we have
\begin{align*}
&\forall\Delta_1>0,\ R\lp 1,2;\Delta_1\rp \le R\lp \emptyset;\Delta_1\rp\\
\iff &R\lp 1,2;\Delta_1=\infty\rp \le R\lp \emptyset;\Delta_1=\infty\rp\\
\iff & \log\lp1+g_1^2+g_2^2\rp - \log\lp\frac{1+\Delta_2}{\Delta_2}\rp\\
&\quad \le \log\lp 1 + \frac{h_2^2}{1+\Delta_2}\rp\\
\iff &\Delta_2 \le \frac{1+h_2^2}{g_1^2+g_2^2}
\end{align*}

Hence, if $\Delta_2 \le \frac{1+h_2^2}{g_1^2+g_2^2}$, then $\Delta_1^* = \Delta_1^*(1,2;2):=$ the intersection of $R\lp 1,2;\Delta_1\rp$ and $R\lp 2;\Delta_1\rp$:
\begin{align*}
\Delta_1^*(1,2;2) := \frac{(1+g_2^2)(1+h_1^2)}{g_1^2}
\end{align*}

Otherwise, we use the monotonicity of these curves to arrive at
$
\Delta_1^* = \min\lbp \Delta_1^*(1,2;\emptyset), \Delta_1^*(1,2;2)\rbp
$,
where $\Delta_1^*(1,2;\emptyset):=$ the intersection of $R\lp 1,2;\Delta_1\rp$ and $R\lp\emptyset;\Delta_1\rp$:
\begin{align*}
\Delta_1^*(1,2;\emptyset) := \frac{(1+h_1^2)\Delta_2+(1+h_1^2+h_2^2)}{(g_1^2+g_2^2)\Delta_2-(1+h_2^2)}
\end{align*}

Below we derive the necessary and sufficient condition for $\Delta_1^*(1,2;\emptyset) \ge \Delta_1^*(1,2;2)$:
\begin{align*}
&\Delta_1^*(1,2;\emptyset) \ge \Delta_1^*(1,2;2)\\
\iff & g_1^2\lp(1+h_1^2)\Delta_2 + (1+h_1^2+h_2^2)\rp\\
&\quad \ge (1+g_2^2)(1+h_1^2)(g_1^2+g_2^2)\lp \Delta_2 - (1+h_2^2)\rp\\
\iff &\Delta_2\le \frac{(1+g_1^2+g_2^2)(1+h_1^2+h_2^2) + (1+g_2^2)h_1^2h_2^2}{g_2^2(1+g_1^2+g_2^2)}\\
&=\delta_1
\end{align*}

The following claim concludes the discussion of this case.
\begin{claim}\label{claim_Order}
For any nonzero $\{h_1,h_2,g_1,g_2\}$,
\begin{align*}
\frac{1+h_2^2}{g_1^2+g_2^2} < \delta_1 < \delta_2
\end{align*}
\end{claim}

Therefore, for $0<\Delta_2 \ge \delta_1$, $\Delta_1^* = \Delta_1^*(1,2;2)$ is the optimal solution, and $\Omega^* = \{1\}$ or $\Omega^* = \{2\}$. For $\delta_1\le\Delta_2 < \delta_2$, $\Delta_1^* = \Delta_1^*(1,2;\emptyset)$ is the optimal solution, and $\Omega^* = \{1,2\}$ or $\Omega^* = \emptyset$.

Combining the above two cases for $\Delta_2 \ge \delta_2$ and $\Delta_2 < \delta_2$, we have the complete characterization of $R_{\text{QMF,G}}^*(\Delta_2)$.
%
%

\subsection{Proof of Claim~\ref{claim_Dominance}}
Assume the contrary, that $\Delta_1^*(1;2) \le \Delta_1^*(1;\emptyset)$. After some manipulations, we have
\begin{align*}
&\Delta_1^*(1;2) \le \Delta_1^*(1;\emptyset) \iff\\
& \lbp\begin{array}{l} (1+h_1^2)g_2^2(1+g_1^2)\Delta_2^2\\
+ \lb\begin{array}{l} (1+h_1^2+h_2^2)(g_2^2-g_1^2)\\ + (g_1^2g_2^2-1)(1+h_1^2)(1+h_2^2)\end{array}\rb \Delta_2\\
- (1+g_1^2)(1+h_2^2)(1+h_1^2+h_2^2) \end{array}\rbp \le 0
\end{align*}

Denote this quadratic function of $\Delta_2$ by $f(\Delta_2)$. Now lets plug in $g_2^2\Delta_2 \ge (1+g_1^2)(1+h_2^2)$ to give a lower bound on $f(\Delta_2)$:
\begin{align*}
&f(\Delta_2)\\
&\ge
g_1^2\lb (1+h_1^2)(1+h_2^2)(2+g_1^2+g_2^2) - (1+h_1^2+h_2^2)\rb \Delta_2\\
& > 0
\end{align*}
which leads to contradiction.

\subsection{Proof of Claim~\ref{claim_Order}}
First of all,
\begin{align*}
&\delta_1 < \delta_2 \\
\iff &\frac{(1+g_1^2+g_2^2)(1+h_1^2+h_2^2) + (1+g_2^2)h_1^2h_2^2}{g_2^2(1+g_1^2+g_2^2)(1+h_1^2)}\\
& < \frac{(1+g_1^2)(1+h_2^2)}{g_2^2}\\
\iff & (1+g_1^2+g_2^2)(1+h_1^2+h_2^2) + (1+g_2^2)h_1^2h_2^2\\
& < (1+g_1^2+g_2^2)(1+h_1^2)(1+g_1^2)(1+h_2^2)\\
\iff & (1+g_2^2)h_1^2h_2^2 < (1+g_1^2+g_2^2)(1+g_1^2)h_1h_2^2\\
&\hspace{72pt} + g_1^2(1+g_1^2+g_2^2)(1+h_1^2+h_2^2)
\end{align*}
which is obviously true since $(1+g_2^2) < (1+g_1^2+g_2^2) < (1+g_1^2+g_2^2)(1+g_1^2)$.

Second,
\begin{align*}
&\frac{1+h_2^2}{g_1^2+g_2^2} < \delta_1\\
\iff & g_2^2(1+g_1^2+g_2^2)(1+h_1^2)(1+h_2^2)\\
&< (g_1^2+g_2^2)(1+g_1^2+g_2^2)(1+h_1^2+h_2^2)\\
&\quad + (g_1^2+g_2^2)(1+g_2^2)h_1^2h_2^2\\
\iff & g_2^2(1+g_1^2+g_2^2)h_1^2h_2^2 \\
&< g_1^2(1+g_1^2+g_2^2)(1+h_1^2+h_2^2) \\
&\quad + (g_1^2+g_2^2)(1+g_2^2)h_1^2h_2^2,
\end{align*}
which is obviously true since $g_2^2(1+g_1^2+g_2^2) < (g_1^2+g_2^2)(1+g_2^2)$.

\section{Proof of Theorem~\ref{thm_2Relay}}\label{app_Pf_thm_2Relay}

By the characterization in Lemma~\ref{lem_2Relay}, \eqref{eq_2Relay_G_Opt} can be solved by finding the optimal solution in each of the above three ranges $\mcal{I}_1,\mcal{I}_2,\mcal{I}_3$ of $\Delta_2$ analytically, and then find the maximum of these three.

For $\Delta_2\in\mcal{I}_3$, note that $\Delta_1^*$ is an increasing function of $\Delta_2$ and that $R\lp\emptyset;\Delta_1,\Delta_2\rp$ decreases when both $\Delta_1$ and $\Delta_2$ increase. Hence we conclude that $R_{\text{QMF,G}}^*(\Delta_2)$ is a decreasing function in this range. Hence
\begin{align*}
\Delta_2^*  &= \delta_2 &
\Delta_1^* &= \frac{(1+h_1^2)\delta_2 + (1+h_1^2+h_2^2)}{g_1^2(\delta_2+(1+h_2^2))}
\end{align*}
in this range.

For $\Delta_2\in\mcal{I}_1$, note that $\Delta_1^*$ does not depend on $\Delta_2$ and that $R\lp1,2;\Delta_1,\Delta_2\rp$ increases when $\Delta_2$ increases. Hence we conclude that $R_{\text{QMF,G}}^*(\Delta_2)$ is an increasing function in this range. Hence
\begin{align*}
\Delta_2^* &=  \delta_1 &
\Delta_1^* &= \frac{(1+g_2^2)(1+h_1^2)}{g_1^2}
\end{align*}
in this range.

For $\Delta_2\in\mcal{I}_2$, unlike the previous two cases, $R_{\text{QMF,G}}^*(\Delta_2)$ may not be monotone in this case. 
Here
\begin{align*}
&R_{\text{QMF,G}}^*(\Delta_2) = \log\lp 1+\frac{h_1^2}{1+\Delta_1^*} + \frac{h_2^2}{1+\Delta_2}\rp\\
&=\log\lp1+g_1^2+g_2^2\rp - \log\lp\frac{1+\Delta_1^*}{\Delta_1^*}\rp - \log\lp \frac{1+\Delta_2}{\Delta_2}\rp
\end{align*}
The derivative of the above function with respect to $\Delta_2$ has the same sign as the quadratic function
$
q(\Delta_2):=A\Delta_2^2 + B\Delta_2 + C
$,
where
\begin{align*}
A &= h_1^2(1+h_1^2) - h_2^2(1+h_1^2+g_1^2+g_2^2) \\
B &= 2h_1^2(1+h_1^2)  \hspace{36pt}
C := h_1^2(1+h_1^2+h_2^2)
\end{align*}
If $A \ge 0$, then the above quadratic is always positive, implying that $R_{\text{QMF,G}}^*(\Delta_2)$ is an increasing function. Therefore
\begin{align*}
\Delta_2^* &= \delta_2 \\
\Delta_1^* &= \frac{(1+h_1^2)\delta_2+(1+h_1^2+h_2^2)}{(g_1^2+g_2^2)\delta_2-(1+h_2^2)}\\
&= \frac{(1+h_1^2)\delta_2 + (1+h_1^2+h_2^2)}{g_1^2(\delta_2+(1+h_2^2))}
\end{align*}
is the optimal solution.

If $A < 0$, since $q(0)=C >0$ and $q(\infty) < 0$, it has a only one positive root
\begin{align*}
\delta_3 := \frac{-B-\sqrt{B^2-4AC}}{2A}
\end{align*}

If $\delta_3 \in\mcal{I}_3$ then 
\begin{align*}
\Delta_2^* &= \delta_3 &
\Delta_1^* &= \frac{(1+h_1^2)\delta_3+(1+h_1^2+h_2^2)}{(g_1^2+g_2^2)\delta_3-(1+h_2^2)}
\end{align*}
is the optimal solution.

If $\delta_3 < \delta_1$, that is, $\delta_3\in\mcal{I}_1$, then $q(\Delta_2) < 0$ for $\Delta_2\in\mcal{I}_3$ and hence $R_{\text{QMF,G}}^*(\Delta_2)$ is a decreasing function of $\Delta_2$ in this range. Therefore,  
\begin{align*}
\Delta_2^* &= \delta_1\\
\Delta_1^* &= \frac{(1+h_1^2)\delta_1+(1+h_1^2+h_2^2)}{(g_1^2+g_2^2)\delta_1-(1+h_2^2)}\\
&= \frac{(1+g_2^2)(1+h_1^2)}{g_1^2}
\end{align*}
is the optimal solution.

If $\delta_3 \ge \delta_2$, that is, $\delta_3\in\mcal{I}_3$, then $q(\Delta_2) > 0$ for $\Delta_2\in\mcal{I}_3$ and hence $R_{\text{QMF,G}}^*(\Delta_2)$ is an increasing function of $\Delta_2$ in this range. Therefore, 
\begin{align*}
\Delta_2^* &= \delta_2 \\
\Delta_1^* &= \frac{(1+h_1^2)\delta_2+(1+h_1^2+h_2^2)}{(g_1^2+g_2^2)\delta_2-(1+h_2^2)}\\
&= \frac{(1+h_1^2)\delta_2 + (1+h_1^2+h_2^2)}{g_1^2(\delta_2+(1+h_2^2))}
\end{align*}
is the optimal solution.

\section{Proof of Lemmas~\ref{lem1_NRelay},\ref{lem2_NRelay} and \ref{lem3_NRelay}}\label{app_Pf_lem1_NRelay}
\subsection{Proof of Lemma \ref{lem1_NRelay}}
Without loss of generality, let $N \geq i > j \geq 0$.
Note that
\begin{align*}
\begin{array}{l}
\log_2 \lbp \lp1 + \frac{(N-i)h^2}{1+\Delta}\rp\lp1+ig^2\rp\lp\frac{\Delta}{1+\Delta}\rp^i \rbp \\ = \log_2 \lbp \lp1 + \frac{(N-j)h^2}{1+\Delta}\rp\lp1+jg^2\rp\lp\frac{\Delta}{1+\Delta}\rp^j \rbp
\end{array}
\end{align*}
if and only if $f(\Delta)=0$, where
\begin{align*}
\begin{array}{l}
f(\Delta) = \\
\lp 1-\frac{1+jg^2}{1+ig^2} \rp\Delta^{i-j+1} +\\
\lp 1+(N-i)h^2  - \lp \frac{1+jg^2}{1+ig^2} \rp \lp 1+(N-j)h^2 + i-j \rp \rp \Delta^{i-j} \\
- \lp \frac{1+jg^2}{1+ig^2} \rp \sum_{p = 1}^{i-j-1} \lbp {{i-j}\choose p} \lp 1+(N-j)h^2 \rp + {{i-j}\choose{p-1}} \rbp \Delta^p\\
-\lp \frac{1+jg^2}{1+ig^2} \rp \lp 1+(N-j)h^2 \rp
\end{array}
\end{align*}

We note that in $f(\Delta)$, the coefficient of $\Delta^{i-j+1}$ is positive and that the coefficients of $\Delta^p_{\{p \in [0:i-j-1]\}}$ are negative. The coefficient of $\Delta^{i-j}$ may be positive or negative, depending on the channel configurations. Either way, the number of \emph{sign changes} of the coefficients of $f(\Delta)$ when written in descending order of powers is exactly $1$. By the Descartes' sign scheme, the number of positive roots of such a polynomial equation is given by $\alpha-2m$, where $\alpha$ is the number of sign changes and $m$ is a positive integer. Since $f(\Delta)$ has \emph{exactly} $1$ sign change, $f(\Delta)$ has \emph{exactly} $1$ positive root, which proves the lemma.

\subsection{Proof of Lemma \ref{lem2_NRelay}}

For $i \in [0:N-1]$, $R_i(\Delta) - R_{i+1}(\Delta)$ can be simplified as,
\begin{align*}
\begin{array}{l}
R_i(\Delta) - R_{i+1}(\Delta) = \\
\log_2 \lp \frac{1+\Delta}{\Delta}\rp + \log_2 \lp \frac{1 + \Delta + (N-i)h^2}{1 + \Delta + (N-i-1)h^2} \rp + \log_2 \lp \frac{1 + ig^2}{1 + (i+1)g^2} \rp
\end{array}
\end{align*}
From the above, we note that for finite $h$ and $g$, $\lim_{\Delta\rightarrow 0}\lbp R_i(\Delta) - R_{i+1}(\Delta) \rbp = +\infty$ due to the presence of the term $\log \lp \frac{1+\Delta}{\Delta}\rp$ in the expression for $R_i(\Delta) - R_{i+1}(\Delta)$, which proves the lemma.

\subsection{Proof of Lemma \ref{lem3_NRelay}}

We define the sequence of polynomials $f_m(\Delta) = \frac{(1 +\Delta)^{m+2}}{\Delta^m}\lp 2^{R_{m+1}(\Delta)} - 2^{R_m(\Delta)} \rp $. Further, let $\Delta^{*}_m$ denote the unique positive root of $f_m(\Delta)$ (which is also the unique positive solution of $R_m(\Delta) = R_{m+1}(\Delta)$). Then, $\Delta^{*}_m$ satisfies
\begin{align*}
\begin{array}{l}
g^2{\Delta^*_m}^2\\
+ \lp g^2 \lp 1 + (N-m)h^2 \rp - 1 - h^2 - g^2h^2 + 2mg^2 \rp{\Delta^*_m} \\
- \lp 1+mg^2 \rp \lp 1 + (N-m)h^2 \rp = 0
\end{array}
\end{align*}
For the required condition to hold true, i.e, $\Delta^{*}_m \leq \Delta^{*}_{m+1}$, we must have $f_{m+1}(\Delta^{*}_m) < 0$, since $R_{m+1}(\Delta) > R_{m+2}(\Delta)$ for $\Delta \in (0,\Delta^{*}_{m+1})$ by the preceding lemma. For $f_{m+1}(\Delta^{*}_m) < 0$,
\begin{align*}
\Delta^{*}_m \geq \frac{h^2\lp1 + (m+1)g^2\rp - g^2\lp1 + (N-m)h^2\rp}{g^2(2 + h^2)} = \Delta_t. 
\end{align*}
What now remains to be shown in order for the lemma to hold is that $f_m(\Delta_t)<0$ for all values of $h,g,N$ and all $m \in [0:N-1]$. Substituting the value of $\Delta_t$, we obtain,
\begin{align*}
\begin{array}{l}
f_m(\Delta_t) = -\frac{h^6}{(2+h^2)^2}\lb \underbrace{\begin{array}{l}m^2g^2 + m\lp2+(3-2N)g^2\rp+\\g^2(N-1)^2-(N-2)+\frac{1}{g^2}\end{array}}_{\Gamma}\rb -\\
\frac{h^4}{(2+h^2)^2}\lb N+2(m+1)+g^2(N^2+2m-1)+\frac{2}{g^2}\rb -\\
\frac{h^2}{(2+h^2)^2}\lb2N+3+4m+g^2\lp4m^2+4m+(2N-1)\rp+\frac{2}{g^2}\rb\\
-\frac{1}{(2+h^2)^2}\lb2+g^2\rb
\end{array}
\end{align*}

From the above expression for $f_m(\Delta_t)$, it is clear that $\Gamma>0$ is a sufficient condition for $f_m(\Delta_t)<0$. $\Gamma$ can be viewed as a quadratic in $m$ with a discriminant $D_{\Gamma} = (5-4N)g^4+4(1-N)g^2$. Since the coefficient of $m^2$ is positive, the necessary and sufficient condition for $\Gamma>0,\ \forall m,g,h,N$ is that $D_{\Gamma}<0$, which is satisfied for $N \geq 2$. This proves the lemma.


\begin{thebibliography}{1}

\bibitem{ZhengTse_03}
L.~Zheng and D.~N.~C. Tse, ``Diversity and multiplexing: a fundamental tradeoff in multiple-antenna channels,'' {\em IEEE Transactions on Information
  Theory}, vol.~49, pp.~1073--1096, May 2003.

\bibitem{AvestimehrDiggavi_09}
A.~S. Avestimehr, S.~N. Diggavi, and D.~N.~C. Tse, ``Wireless network
  information flow: A deterministic approach,'' {\em IEEE Transactions on
  Information Theory}, vol.~57, pp.~1872--1905, April 2011.

\bibitem{PawarAvestimehr_08}
S.~Pawar, A.~S. Avestimehr, and D.~N.~C. Tse, ``Diversity-multiplexing tradeoff
  of the half-duplex relay channel,'' {\em Proceedings of Allerton Conference
  on Communication, Control, and Computing}, pp.~27--33, September 2008.

\bibitem{LimKim_10}
S.~H. Lim, Y.-H. Kim, A.~A. $\textrm{El Gamal}$, and S.-Y. Chung, ``Noisy
  network coding,'' {\em IEEE Transactions on Information Theory}, vol.~57,
  pp.~3132--3152, May 2011.

\bibitem{CoverElGamal_79}
T.~M. Cover and A.~A. $\textrm{El Gamal}$, ``Capacity theorems for the relay
  channel,'' {\em IEEE Transactions on Information Theory}, vol.~25,
  pp.~572--584, September 1979.

\bibitem{NazerGastpar_11}
B.~Nazer and M.~Gastpar, ``Compute-and-forward: harnessing interference through structured codes ,'' {\em IEEE Transactions on Information Theory}, vol.~57, pp.~6463--6486, October 2011.

\bibitem{Schein}
B. E. Schein, ``Distributed Coordination in Network Information Theory'',
{\em Ph.D thesis}, Massachusetts Institute of Technology, Sept. 2001.

\bibitem{SWF_Allerton}
A.~Sengupta, I-H.~Wang, C.~Fragouli, ``Optimizing quantize-map-and-forward in slow fading relay networks,''
{\em Proceedings of Allerton Conference on Communication, Control, and Computing}, pp.~1928--1934, October 2012.

\bibitem{SWF_ITW}
A.~Sengupta, I-H.~Wang, C.~Fragouli, ``Optimizing quantize-map-and-forward relaying for Gaussian diamond networks,''
{\em Proceedings of IEEE Information Theory Workshop (ITW)}, pp.~381--385, September 2012.

\bibitem{Chern_logN}
B.~Chern and A.~Ozgur, ``Achieving the capacity of the N-relay Gaussian diamond network within log N bits,'' {\em Proceedings of IEEE Information Theory Workshop (ITW)}, pp.~377--380, September 2012.

\bibitem{YukselErkip_07}
M.~Yuksel and E.~Erkip, ``Multiple-antenna cooperative wireless systems: a diversity-multiplexing tradeoff perspective, '' {\em IEEE Transactions on Information
  Theory}, vol.~53, pp.~3371--3393, October 2007.

\bibitem{Azarian_05}
K.~Azarian, H.~El~Gamal, and P.~Schniter, ``On the achievable diversity-multiplexing tradeoff
in half-duplex cooperative channels,'' {\em IEEE Transactions on Information
  Theory}, vol.~51, pp.~4152--4172, December 2005.

\bibitem{PrasadVaranasi_10}
N.~Prasad and M.~K. Varanasi, ``High performance static and dynamic
cooperative communication protocols for the half duplex fading relay channel,'' {\em IEEE Transactions on Wireless Communications}, vol.~9, pp.~328--337, January 2010.

\bibitem{BletsasShin_07}
A.~Bletsas, H.~Shin, and M.~Z. Win, ``Cooperative communications with
outage-optimal opportunistic relaying,'' {\em IEEE Transactions on Wireless Communications}, vol.~6, pp.~3450--3460, September 2007.

\bibitem{HwangKo_08}
K.-S.~Hwang, Y.-C.~Ko, and M.-S.~Alouini, ``Outage probability of cooperative diversity systems with opportunistic relaying based on decode-and-forward,'' {\em IEEE Transactions on Wireless Communications}, vol.~7, pp.~5100--5107, December 2008.

\bibitem{ConneKim_10}
C.~Conne and I.-M. Kim, ``Outage probability of multi-hop amplify-and-forward relay systems,'' {\em IEEE Transactions on Wireless Communications}, vol.~9, pp.~1139--1149, March 2010.

\bibitem{TourkiYang_11}
K.~Tourki, H.-C.~Yang, and M.-S.~Alouini, ``Accurate outage analysis of incremental decode-and-forward opportunistic relaying,'' {\em IEEE Transactions on Wireless Communications}, vol.~10, pp.~1021--1025, April 2011.

\bibitem{YuLee_12}
H.~Yu, I.-H.~Lee, and G.~L. St\"{u}ber, ``Outage probability of decode-and-forward cooperative relaying systems with
co-channel interference,'' {\em IEEE Transactions on Wireless Communications}, vol.~11, pp.~266--274, Janurary 2012.

\bibitem{YaoKim_13}
S.~Yao, T.~T.~Kim, M. Skoglund, and H.~V.~Poor, ``Half-duplex relaying over slow fading channels
based on quantize-and-forward,'' {\em IEEE Transactions on Information
  Theory}, vol.~59, pp.~860--872, Feburuay 2013.

\bibitem{HouKramer_13}
J.~Hou and G.~Kramer, ``Short message noisy network coding,'' {\em Submitted to IEEE Transactions on Information
  Theory}, April 2013. \url{http://arxiv.org/abs/1304.1692}

\bibitem{Lattice_QMF}
A.~Ozgur and S~N. Diggavi,
\newblock {``Approximately achieving Gaussian relay network capacity with
lattice codes''},
\newblock {\em Proc. of IEEE ISIT 2010, Austin, Texas}, pp. 669-673, June
2010.

\end{thebibliography}
\end{document}